\documentclass[a4paper,10pt,twoside]{article}
\pdfoutput=1

\makeatletter

\usepackage{amsmath}
\usepackage{amsthm}

\usepackage[english]{babel}

\usepackage{color}

\usepackage{enumitem}

\usepackage[T1]{fontenc}

\usepackage[unicode,colorlinks,pdfusetitle,final]{hyperref}

\usepackage[utf8]{inputenc}

\usepackage{mathtools}

\usepackage[numbers,sort]{natbib}

\usepackage{textcomp}

\usepackage{url}

\IfFileExists{MinionPro.sty}{
  \usepackage[lf,medfamily,italicgreek,openg]{MinionPro}
}{
  \usepackage{lmodern}
  \usepackage{upgreek}
  \usepackage{amssymb}
  \usepackage{amsfonts}

}

\IfFileExists{MyriadPro.sty}{
  \usepackage[lf,medfamily,onlytext]{MyriadPro}

  \usepackage[sf,bf,small]{titlesec}
  \usepackage[labelfont={sf,bf},labelsep=period]{caption}

  \let\maybesf\sffamily
}{
  \usepackage{helvet}

  \usepackage[bf,small]{titlesec}
  \usepackage[labelfont={bf},labelsep=period]{caption}

  \let\maybesf\rmfamily
}

\usepackage[final]{microtype}

\titlelabel{\thetitle.\space}

\numberwithin{equation}{section}

\newcounter{and}
\newcommand{\institute}[1]{\newcommand{\@institute}{#1}}

\newcommand{\email}[1]{\href{mailto:#1}{#1}}

\renewcommand{\maketitle}{
  { 
    \raggedright
    \LARGE
    \noindent
    \bfseries
    \maybesf
    \@title
    \par
  }

  \vspace{1.5\baselineskip}

  { 
    \raggedright
    \renewcommand{\and}{
      \unskip, \ignorespaces
    }
    \noindent\ignorespaces\@author\par
  }

  \vspace{0.5\baselineskip}

  { 
    \raggedright
    \small
    \renewcommand{\and}{
      \par\stepcounter{and}
      \hangindent .8em\noindent
      \hbox to .8em{\smash{$^{\arabic{and}}$}}\ignorespaces
    }
    \ifnum\value{and}=0
      \noindent
    \else
      \hangindent .8em\noindent
      \hbox to .8em{\smash{$^{\arabic{and}}$}}\ignorespaces
    \fi
    \ignorespaces\@institute\par
  }

  \vspace{1.0\baselineskip}
 
  { 
    \raggedright
    \noindent\ignorespaces\@date\par
  }

}

\renewenvironment{abstract}{
  \addvspace{1.5\baselineskip}
  \topsep=0pt\partopsep=0pt
  \trivlist\item[\hskip\labelsep\bfseries\maybesf Abstract.]
}{}

\newenvironment{acknowledgments}{
  \addvspace{1.5\baselineskip}
  \topsep=0pt\partopsep=0pt
  \trivlist\item[\hskip\labelsep\bfseries\maybesf Acknowledgments.]
}{}

\def\bbR{\mathbb{R}}
\def\bbC{\mathbb{C}}
\def\bbN{\mathbb{N}}

\def\bbM{\mathbb{M}}

\def\M{\mathcal{M}}
\def\N{\mathcal{N}}
\def\A{\mathcal{A}}
\def\F{\mathcal{F}}

\def\H{\mathcal{H}}
\def\D{\mathcal{D}}

\def\G{\mathcal{G}}
\def\R{\mathcal{R}}

\def\OO{\mathcal{O}}

\def\Imm{\mathrm{Im}\,}
\def\Ree{\mathrm{Re}\,}

\def\pp{\mathrm{pp}}
\def\sd{\mathrm{sd}}
\def\ms{\mathrm{ms}}

\def\supp{\mathrm{supp}}

\def\WF{\mathrm{WF}}
\def\loc{\mathrm{loc}}
\def\reg{\mathrm{reg}}
\def\balpha{{\boldsymbol{\alpha}}}
\def\bbeta{{\boldsymbol{\beta}}}
\def\1{\mathbbm{1}}

\def\beq{\begin{equation}}
\def\eeq{\end{equation}}

\newcommand{\wick}[1]{:\!{#1}\!:}

\theoremstyle{plain}
\newtheorem{theo}{Theorem}[section]
\newtheorem{lem}[theo]{Lemma}
\newtheorem{propo}[theo]{Proposition}

\theoremstyle{definition}
\newtheorem{defi}[theo]{Definition}

\newtheorem{rem}[theo]{Remark}

\definecolor{hypercolor}{rgb}{0.1,0.2,0.6}

\hypersetup{    
 unicode=false,     
 pdftoolbar=true,   
 pdfmenubar=true,   
 pdffitwindow=true, 
 pdfstartview={FitH},
 pdftitle={Draft},   
 pdfauthor={Antoine G\'er\'e, Thomas-Paul Hack, Nicola Pinamonti},    
 pdfsubject={Mathematical Physics},  
 pdfcreator={LaTeX}, 
 pdfproducer={pdfTex},
 pdfkeywords={Perturbative Algebraic Quantum Field Theory}, 
 pdfnewwindow=true, 
 colorlinks=true,
 linkcolor=hypercolor,
 urlcolor=hypercolor,
 citecolor=hypercolor,
 filecolor=hypercolor,        
}

\begin{document}

\title{An analytic regularisation scheme on curved spacetimes with applications to cosmological spacetimes}

\author{
  Antoine G\'er\'e$^{1,a}$\and
  Thomas-Paul Hack$^{1,a}$\and
  Nicola Pinamonti$^{1,2,c}$
}

\institute{
  $^1$Dipartimento di Matematica, Universit\`a  di Genova, Via Dodecaneso 35,
  I-16146 Genova, Italy.\\
  $^2$Istituto Nazionale di Fisica Nucleare, Sezione di Genova, Via Dodecaneso, 33 I-16146 Genova, Italy.\\
  E-Mail: \email{$^a$gere@dima.unige.it}, \email{$^b$hack@dima.unige.it}, \email{$^c$pinamont@dima.unige.it}
}

\date{\today}

\maketitle


\begin{abstract}We develop a renormalisation scheme for time--ordered products in interacting field theories on curved spacetimes which consists of an analytic regularisation of Feynman amplitudes and a minimal subtraction of the resulting pole parts. This scheme is directly applicable to spacetimes with Lorentzian signature, manifestly generally covariant, invariant under any spacetime isometries present and constructed to all orders in perturbation theory. Moreover, the scheme captures correctly the non--geometric state--dependent contribution of Feynman amplitudes and it is well--suited for practical computations. To illustrate this last point, we compute explicit examples on a generic curved spacetime, and demonstrate how momentum space computations in cosmological spacetimes can be performed in our scheme. In this work, we discuss only scalar fields in four spacetime dimensions, but we argue that the renormalisation scheme can be directly generalised to other spacetime dimensions and field theories with higher spin, as well as to theories with local gauge invariance. 
\end{abstract}


\tableofcontents


\section{Introduction}

In the perturbative construction of models in quantum field theory on curved spacetimes one encounters time--ordered products of field polynomials which are a priori ill--defined due to the appearance of UV divergences. Several renormalisation schemes which deal with these divergences in the presence of non--trivial spacetime curvature have been discussed in the literature, such as for example local momentum space methods \cite{Bunch:1981er}, dimensional regularisation in combination with heat kernel techniques \cite{Luscher:1982wf, Toms:1982af}, differential renormalisation \cite{Comellas:1994da,Prange:1997iy}, zeta--function renormalisation \cite{Bilal:2013iva}, generic Epstein--Glaser renormalisation \cite{Brunetti-Fredenhagen:2000, Hollands:2001nf, Hollands:2004yh}, and, on cosmological spacetimes, Mellin--Barnes techniques \cite{Hollands:2010pr} and dimensional regularisation with respect to the comoving spatial coordinates \cite{Baacke:2010bm}. 

Some of these schemes, such as heat kernel approaches, zeta--function techniques and local momentum space methods  are based on constructions which are initially only well--defined for spacetimes with Euclidean signature. These constructions can be partly transported to general Lorentzian spacetimes by local Wick--rotation techniques developed in \cite{Moretti:1999fb}. However, whereas the Feynman propagator is essentially unique on Euclidean spacetimes, this is not the case on Lorentzian spacetimes where this propagator has a non--unique contribution depending on the quantum state of the field model. Consequently, the Euclidean renormalisation techniques, and the numerous practical computations already performed by means of these methods -- see for example the monographs  \cite{Birrell:1982ix, Buchbinder:1992rb, Parker:2009uva} -- are able to capture the correct divergent and geometric parts of Feynman amplitudes, but a priori not their non--geometric and state--dependent contributions.

A renormalisation scheme which is directly applicable to curved spacetimes with Lorentzian signature has been developed in \cite{Brunetti-Fredenhagen:2000, Hollands:2001nf, Hollands:2004yh} in the framework of algebraic quantum field theory. This scheme implements ideas of \cite{EG} and \cite{Steinmann} and is based on microlocal techniques which replace the momentum space methods available in Minkowski spacetime and have been introduced to quantum field theory in curved spacetime by the seminal work \cite{Radzikowski}. However, although the generalised Epstein--Glaser scheme developed in \cite{Brunetti-Fredenhagen:2000, Hollands:2001nf, Hollands:2004yh} is conceptually clear and mathematically rigorous, it is not easily applicable in practical computations. On the other hand, Lorentzian schemes which are better suited for this purpose have not been developed to all orders in perturbation theory \cite{Comellas:1994da,Prange:1997iy}, are tailored to specific spacetimes \cite{Hollands:2010pr} or are not manifestly covariant \cite{Baacke:2010bm}.

Motivated by this, we develop a renormalisation scheme for time--ordered products in interacting field theories on curved spacetimes which is directly applicable to spacetimes with Lorentzian signature, manifestly generally covariant, invariant under any spacetime isometries present and constructed to all orders in perturbation theory. Moreover, the scheme captures correctly the non--geometric state--dependent contribution of Feynman amplitudes and it is well--suited for practical computations. In this work, we discuss only scalar fields in four spacetime dimensions, but we shall argue that the renormalisation scheme can be directly generalised to other spacetime dimensions and field theories with higher spin, as well as to theories with local gauge invariance. Our analysis will take place in the framework of perturbative algebraic quantum field theory (pAQFT) \cite{Brunetti-Fredenhagen:2000, Hollands:2001nf, Hollands:2004yh, BDF,FredenhagenRejzner,FredenhagenRejzner2} which is a conceptually clear framework in which fundamental physical properties of perturbative interacting models on curved spacetimes can be discussed. However, we will make an effort to review how the formulation of pAQFT is related to the more standard formulation of perturbative QFT.

The renormalisation scheme we propose is inspired by the works \cite{Keller,dfkr} which deal with perturbative QFT in Minkowski spacetime. In these works, the authors introduce an analytic regularisation of the position--space Feynman propagator in Minkowski spacetime which is similar to the one discussed in \cite{BolliniGiambiagi}. Based on this, time--ordered products are constructed recursively by an Epstein--Glaser type procedure and it is shown that this recursion can be resolved by a position--space forest formula similar to the one of Zimmermann used in BPHZ renormalisation in momentum space. 

In order to extend the scheme proposed in \cite{dfkr} to curved spacetimes, and motivated by \cite{BolliniGiambiagi} and by the form of Feynman propagators on curved spacetimes, we introduce an analytic regularisation $\Delta^{(\alpha)}_F$ of a Feynman propagator $\Delta_F$ by
\[
\Delta^{(\alpha)}_F := \lim_{\epsilon\to 0^+}\frac{1}{8\pi^2}\left(\frac{u}{(\sigma+i\epsilon )^{1+\alpha}} + \frac{v}{\alpha} \left(1-\frac{1}{(\sigma+i\epsilon )^{\alpha}}\right)\right)+w\,,
\]
where $u$, $v$ and $w$ are the so--called Hadamard coefficients and $\sigma$ is $1/2$ times the squared geodesic distance. This analytic regularisation, namely the construction of certain distributions by means of powers of the geodesic distance, is reminiscent of the use of Riesz distributions to define advanced and retarded Greens functions on Minkowski spacetime. A careful discussion of Riesz distributions and their extension to the curved case is presented in \cite{BGP}. The regularisation we use is loosely related to dimensional regularisation because the leading singularity of a Feynman propagator in $N$ spacetime dimensions is proportional to $(\sigma+i\epsilon)^{1-N/2}$, see e.g. \cite[Appendix A]{Moretti}.  A regularisation of the Feynman propagator similar to the one above has recently been discussed in \cite{Dang}. In this work, we shall combine the analytic regularisation of the Feynman propagator with the minimal subtraction scheme encoded in a forest formula of the kind discussed in \cite{Hollands:2010pr, Keller, dfkr} in order to obtain a time--ordered product which satisfies the causal factorisation property, i.e. a product which is indeed ``time--ordered''. In order to prove that the analytically regularised amplitudes constructed out of $\Delta^{(\alpha)}_F$ have the meromorphic structure necessary for the application of the forest formula and in order to show how the corresponding Laurent series can be computed explicitly, we shall make use of generalised Euler operators.
The practical feasibility of the renormalisation scheme shall be demonstrated by computing a few examples.

A large part of our analysis is devoted to demonstrating that the scheme we propose is consistent and well-defined on general curved spacetimes. Readers interested in directly applying our scheme may use Proposition \ref{pr:expose-poles} with \eqref{eq:euler-operator} and \eqref{eq:euler-operator2} in order to compute the Laurent-series of the Feynman amplitudes \eqref{eq:tau-gamma_reg} regularised by means of \eqref{eq:anal-feynman}. The correct order of pole subtractions is encoded in the forest formula \eqref{eq:forset-formula} which is explained prior to its display. A number of examples is discussed in Section \ref{sec_fishsunset}.

In quantum field theory on cosmological spacetimes, i.e. Friedmann--Lema\^itre--Robertson--Walker (FLRW) spacetimes, one usually exploits the high symmetry of these spacetimes in order to evaluate analytical expressions in spatial Fourier space. However, the renormalisation scheme discussed in this work operates on quantities such as the geodesic distance and the Hadamard coefficients, whose explicit position space and momentum space forms are not even explicitly known in FLRW spacetimes. Notwithstanding, we shall devote a large part of this work in order to develop simple methods to evaluate quantities renormalised in our scheme on FLRW spacetimes in momentum space, and we shall illustrate these methods by explicit examples.

The paper is organised as follows. In the next section we present a brief introduction to pAQFT and its connection with the more standard formulation of perturbative QFT. Afterwards we introduce the renormalisation scheme, demonstrate that it is well--defined and analyse its properties in Section 3, where we also illustrate the scheme by computing examples. In the fourth section we demonstrate the applicability of the renormalisation scheme to momentum space computations on cosmological spacetimes. Finally, a few conclusions are drawn in the last section of this paper. Conventions regarding the various propagators of a scalar field theory and a few technical computations are collected in the appendix.

\bigskip


\section{Introduction to pAQFT}

\subsection{Basic definitions}
\label{sec:pAQFT.basic}

Throughout this work, we shall consider four-dimensional globally hyperbolic spactimes $(\M,g)$, where $g$ is a Lorentzian metric whose signature is $(-,+,+,+)$ and we use the sign conventions of \cite{Wald:1984} regarding the definitions of curvature tensors.

We recall the perturbative construction of an interacting quantum field theory on a generic curved spacetime in  the framework of {\bf perturbative algebraic quantum field theory (pAQFT)} recently developed in  \cite{BDF,FredenhagenRejzner,FredenhagenRejzner2} based on earlier work. 
In this construction, the basic object of the theory is an algebra of observables which is realised as a suitable set of functionals on field configurations equipped with a suitable product.
In order to implement the perturbative constructions following the ideas of Bogoliubov and others, the {\bf field configurations} $\phi$ are assumed to be off--shell. Namely, $\phi\in\mathcal{E}(\M)=C^\infty(\M)$ is a smooth function on the globally hyperbolic spacetime $(\M,g)$ and observables are modelled by functionals $F:\mathcal{E}(\M)\to \bbC$ satisfying further properties. In particular all the functional derivatives exist as distributions of compact support, where we recall that the functional derivative of a functional $F$ is defined for all $\psi_1,\ldots,\psi_n\in \D(\M)=C_0^\infty(\M)$ as
\[ 
F^{(n)}(\phi)(\psi_1\otimes \dots \otimes \psi_n) :=  \left.\frac{d^n}{d\lambda_1    \dots d\lambda_n } F(\phi + \lambda_1 \psi_1 +\dots \lambda_n \psi_n)\right|_{\lambda_1 = \dots=\lambda_n=0} \in \mathcal{E}'(\M^n).
\]
The set of these functionals is indicated by $\mathcal{F}$. 
Further regularity properties are assumed for the construction of an algebraic product.  In particular, the set of {\bf local functionals} $\mathcal{F}_\loc\subset \mathcal{F}$ is formed by the functionals whose $n$--th order functional derivatives are supported on the total diagonal $d_n = \{(x,\dots, x) , x\in \M\}\subset \M^n $. Furthermore, their singular directions are required to be orthogonal to $d_n$, namely $\WF(F^{(n)})\subset \{ (x,k) \in T^*\M^n, x\in d_n, k \perp  T d_n \}$ where $\WF$ denotes the wave front set. A generic local functional is a polynomial $P(\phi)(x)$ in $\phi$ and its derivatives integrated against a smooth and compactly supported tensor. The functionals whose functional derivatives are compactly supported smooth functions are instead called {\bf regular functionals} and indicated by $\mathcal{F}_\reg$.

The quantum theory is specified once a product among elements of $\mathcal{F}_\loc$ and a $*-$operation (an involution on $\mathcal{F}$) are given. For the case of free (linear) theories the product can be explicitly given by a {\bf $\star$--product }
\beq\label{def:starH}
F \star_H  G =  \sum_n \frac{\hbar^n}{n!}\left\langle F^{(n)}, H_+^{\otimes n} G^{(n)} \right\rangle\,,
\eeq
where $H_+$ is a Hadamard distribution of the linear theory we are going to quantize, namely a distribution whose antisymmetric part is proportional to the commutator function $\Delta=\Delta_R-\Delta_A$ and whose wave front set satisfies the Hadamard condition, see e.g. \cite{Radzikowski, bfk:1996} for further details and Section \ref{sec_propagators} for our propagator conventions. Owing to the properties of $H_+$, iterated $\star_H$--products of local functionals $F_1 \star_H \dots \star_H F_n$ are well defined and $\star_H$ is associative.

In a normal neighbourhood of $(\M,g)$, a Hadamard distribution $H_+$ is of the form
\begin{equation}\label{eq:hadamard}
H_+(x,y)=\frac{1}{8\pi^2}\left(\frac{u(x,y)}{\sigma_+(x,y)}+v(x,y)\log(M^2 \sigma_+(x,y))\right)+w(x,y),
\end{equation}
where $\sigma_+(x,y)=\sigma(x,y)+i\epsilon (t(x)-t(y))+\epsilon^2/2$ with $t$ a time-function, i.e. a global time-coordinate, $2\sigma(x,y)$ is the squared geodesic distance between $x$ and $y$ and $M$ is an arbitrary mass scale. The Hadamard coefficients $u$ and $v$ are purely geometric and thus state--independent, whereas $w$ is smooth and state--dependent if $H_+(x,y)$ is the two--point function of a quantum state.

For the perturbative construction of interacting models we further need a {\bf time--ordered product} $\cdot_{T_H}$ on local functionals. This product is characterised by {\bf symmetry} and the {\bf causal factorisation property}, which requires that 
\begin{equation}\label{eq:causal-factorisation}
F\cdot_{T_H} G =  F\star_H G\quad    \text{if}\quad F\gtrsim G\,,
\end{equation}
where $F \gtrsim G$ indicates that $F$ is later than $G$, i.e. there exists a Cauchy surface $\Sigma$ of $(\M,g)$ such that $\supp(F) \subset J^+(\Sigma)$ and $\supp(G) \subset J^-(\Sigma)$. However, the causal factorisation fixes uniquely only the time--ordered products among regular functionals, in which case
\beq\label{def:TH}
F \cdot_{T_H}  G =  \sum_n \frac{\hbar^n}{n!}\left\langle F^{(n)}, H_F^{\otimes n} G^{(n)} \right\rangle,
\eeq
where $H_F$ is the time--ordered (Feynman) version of $H_+$, i.e. $H_F=H_++i\Delta_A$ with $\Delta_A$ the advanced propagator of the free theory, cf. Section \ref{sec_propagators}. For local functionals, \eqref{def:TH} is only correct up to the need to employ a non--unique renormalisation procedure, cf. Section \ref{sec:analytic_general}. This renormalisation can be performed in such a way that iterated $\cdot_{T_H}$--products of local functionals $F_1 \cdot_{T_H} \dots \cdot_{T_H} F_n$ are well defined with $\cdot_{T_H}$ being associative. Moreover, $\star_H$--products of such time--ordered products of local functionals are well--defined as well, cf. \cite{Hollands:2001b, BDF,FredenhagenRejzner,FredenhagenRejzner2}. Consequently, we may consider the algebra $\A_0$ $\star_H$--generated by iterated $\cdot_{T_H}$--products of local functionals. This algebra contains all observables of the free theory which are relevant for perturbation theory.

In the perturbative construction of interacting models, namely when the free action is perturbed by a non--linear local functional $V$, the observables associated with the interacting theory are represented on the free algebra $\A_0$ by means of the {\bf Bogoliubov formula}. This is given in terms of the local $S$--matrix, i.e., the time--ordered exponential
\beq\label{def:Smatrix}
S(V)=\sum^\infty_{n=0}\frac{i^n}{n!\hbar^n}\underbrace{V\cdot_{T_H} \cdots \cdot_{T_H} V}_{n \text{ times}}\,,
\eeq
where $V$ is the interacting Lagrangean.  In particular, for every interacting observable $F$ the corresponding representation on the free algebra $\A_0$ is given by
\begin{equation}\label{eq:bogoliubov}
\mathcal{R}_V(F) = S^{-1}(V)\star_H \left(  S(V)\cdot_{T_H} F \right)\,, 
\end{equation}
where $S^{-1}(V)$ is the inverse of $S(V)$ with respect to the $\star_H$--product. The problem in using $\mathcal{R}_V(F)$ as generators of the algebra of interacting observables lies in the construction of the time--ordered product which a priori is an ill--defined operation.

This problem can be solved using ideas which go back to Epstein and Glaser, see e.g. \cite{Brunetti-Fredenhagen:2000}, by means of which the time--ordered product among local functionals is constructed recursively.
The time--ordered products can be expanded in terms of distributions smeared with compactly supported smooth functions which play the role of coupling constants (multiplied by a spacetime--cutoff). At each recursion step the causal factorisation property \eqref{eq:causal-factorisation} permits to construct the distributions defining the time--ordered product up to the total diagonal. 
The extension to the total diagonal can be performed extending the distributions previously obtained without altering the scaling degree towards the diagonal. In this procedure there is the freedom of adding finite local contributions supported on the total diagonal. This freedom is the well known renormalisation freedom. In addition to the properties already discussed, the renormalised time--ordered product is required to satisfy further physically reasonable conditions. We refer to \cite{Hollands:2001b,Hollands:2004yh} for details on these properties and the proof that they can be implemented in the recursive Epstein--Glaser construction.

In spite of the theoretical clarity of this construction, the Epstein--Glaser renormalisation is quite difficult to implement in practise. The aim of this paper is to discuss a renormalisation scheme which is suitable for practical computations.

\subsection{Relation to the standard formulation of perturbative QFT}
\label{sec:relationPAQFT}

In this subsection we outline the relation of the pAQFT framework to the standard formulation of perturbative QFT. As an example, we demonstrate how the two-point (Wightman) function of the interacting field in $\phi^4$ theory on a four--dimensional curved spacetime is computed, where we assume that the quantum state of the interacting field is just the state of free field modified by the interacting dynamics. We further assume that the free field is in a pure and Gaussian Hadamard state. 

Let us recall the relevant formulae in perturbative algebraic quantum field theory where we shall always try to write expressions both in the pAQFT and in the more standard notation, indicating the latter by a $\doteq$. Given a local action $V$, such as $V=\int_\M d^4x \sqrt{-g} \; \frac{\lambda}{4} \phi(x)^4$ in $\phi^4$--theory, the corresponding $S$-matrix, which is loosely speaking the ``$S$-matrix in the interaction picture'', is defined by \eqref{def:Smatrix} and corresponds to $S(V)\doteq T e^{\frac{i}{\hbar} V}$.

The interacting field, i.e. ``the field in the interaction picture'' $\phi_I(x)$, is defined by the Bogoliubov formula
\begin{equation}\label{eq_bogoliubov}
\phi_I(x)=\mathcal{R}_V(\phi(x))=S(V)^{-1}\star_H\left(S(V)\cdot_{T_H} \phi(x)\right)\doteq T(e^{\frac{i}{\hbar} V})^{-1} T(e^{\frac{i}{\hbar} V}\phi(x))\,
\end{equation}
similarly to \eqref{eq:bogoliubov}, where by unitarity $S(V)^{-1}=S(V)^*$. Interacting versions of more complicated expressions in the field, e.g. polynomials at different and coinciding points, are defined analogously. A thorough discussion of the relation between the Bogoliubov formula and the more common formulation of observables in the interaction picture may be found e.g. in \cite[Section 3.1]{Lindner:2013ila}. We only remark that, in the Minkowski vacuum state $\Omega_0$, the expectation value of the Bogoliubov formula can be shown to read (also for more general expressions in the field)

$$\langle \phi_I(x)\rangle_{\Omega_0} \doteq \left\langle T(e^{\frac{i}{\hbar} V})^{-1}T(e^{\frac{i}{\hbar} V}\phi(x))\right\rangle_{\Omega_0}=\frac{\left\langle T(e^{\frac{i}{\hbar} V}\phi(x))\right\rangle_{\Omega_0}}{\left\langle T(e^{\frac{i}{\hbar} V}) \right\rangle_{\Omega_0}}\,,$$
which is the theorem of Gell--Mann and Low, see \cite{Duetsch:1996eh, Duetsch:2000nh} for details.

In the algebraic formulation one usually cuts off the interaction in order to avoid infrared problems by replacing $\lambda\to \lambda f(x)$ with a compactly supported smooth function $f$ and considers the adiabatic limit $f\to 1$ in the end when computing expectation values. As our aim is to compute expectation values in this section, we shall write the results in the adiabatic limit keeping in mind that proving the absence of infrared problems, i.e. the convergence of the spacetime integrals, is non--trivial and may depend on the state of the free field chosen. Note that the so-called ``in--in--formalism'' often used in perturbative QFT on cosmological spacetimes corresponds to considering a cutoff function $f$ of the form $f(t,\vec{x}) = \Theta(t-t_0)$, i.e. $f$ is a step function in time and the parameter $t_0$ corresponds to the time where the interaction is switched on.

Our choice for the quantum state $\Omega$ of the interacting field implies that e.g. the interacting two-point function 
$$\langle \phi_I(x)\star_H\phi_I(y)\rangle_\Omega\doteq \langle \phi_I(x)\phi_I(y)\rangle_\Omega$$
is computed by writing $\phi_I$ in terms of the free field $\phi$ and computing the expectation value of the resulting observable of the free field in the chosen pure,  
Gaussian, Hadamard state of the free field which we may thus denote by the same symbol $\Omega$. 
The interacting vacuum state in Minkowski spacetime is of this form, whereas interacting thermal states in flat spacetime do not belong to this class, as they roughly speaking require to take into account both the change of dynamics and the change of spectral properties induced by $V$ \cite{Fredenhagen:2013cna}.

The functionals in the functional picture of pAQFT correspond to Wick--ordered quantities of the free field in the sense we shall explain now. To this avail we recall the form of the (quantum) $\star_H$--product and (time--ordered) $\cdot_{T_H}$--product in \eqref{def:starH} and \eqref{def:TH} which are defined by means of a Hadamard distribution $H_+$ and its Feynman--version $H_F = H_+ + i \Delta_A$. Up to renormalisation of the time--ordered product, these products computed for the special case of the functional $\phi^2(x)$ give
$$
\phi(x)^2\star_H \phi(y)^2=\phi(x)^2\phi(y)^2+4\hbar \phi(x)\phi(y) H_+(x,y)+2\hbar^2 H^2_+(x,y)\,,
$$
$$
\phi(x)^2\cdot_{T_H}\phi(y)^2=\phi(x)^2\phi(y)^2+4\hbar \phi(x)\phi(y) H_F(x,y)+2\hbar^2 H^2_F(x,y)\,.
$$

This example shows that the $\star_H$--product ( $\cdot_{T_H}$--product) implements the Wick theorem for normal--ordered (time--ordered) fields, and thus the previous formulae can be interpreted in more standard notation as 
$$\wick{\phi(x)^2}_H\,\wick{\phi(y)^2}_H=\wick{\phi(x)^2\phi(y)^2}_H+4\hbar \wick{\phi(x)\phi(y)}_H H_+(x,y)+2\hbar^2 H^2_+(x,y)\,,$$
$$T\left(\wick{\phi(x)^2}_H\,\wick{\phi(y)^2}_H\right)=\wick{\phi(x)^2\phi(y)^2}_H+4\hbar \wick{\phi(x)\phi(y)}_H H_F(x,y)+2\hbar^2 H^2_F(x,y)\,,$$
where
\begin{eqnarray}\label{eq_defalpha}
{\wick{A}}_H\,:=\alpha_{-H_+}(A):=e^{-\hbar\left\langle H_S(x,y),\frac{\delta}{\delta\phi(x)}\otimes \frac{\delta}{\delta\phi(y)}\right\rangle}A\,,\\
H_S(x,y) := \frac12 (H_+(x,y)+H_+(y,x))\,,\notag
\end{eqnarray}
e.g. 
$$\wick{\phi(x)^2}_H=\lim_{x\to y}\left(\phi(x)\phi(y)-H_+(x,y)\right)\,.$$ 
\eqref{eq_defalpha} is a convenient way to encode the combinatorics of normal ordering whereby the exponential series terminates for polynomial functionals such as $A=\phi(x)^2$.

The Wick theorem relates (time--ordered) products of Wick--ordered quantities to sums of Wick--ordered versions of contracted products, where the definition of ``Wick--ordering'' and ``contraction'' are directly related, they both depend on the Hadamard distribution  $H_+$ chosen. Thus, if we choose a particular $H_+$ to define $\star_H$ and $\cdot_{T_H}$ in pAQFT, we immediately fix the interpretation of all functionals in terms of expressions Wick--ordered with respect to $H_+$.

For the algebraic formulation the choice of $H_+$ is not important, indeed choosing a different $H'_+$ with the same properties, one has that $w:=H'_+-H_+=H'_F-H_F$, because the advanced propagator $\Delta_A$ is unique and thus universal. Moreover, $w$ is real, smooth and symmetric and
$$A\star_{H'}B=\alpha_w\left(\alpha_{-w}(A)\star_H\alpha_{-w}(B)\right),\qquad A\cdot_{T_{H'}}B=\alpha_w\left(\alpha_{-w}(A)\cdot_{T_H}\alpha_{-w}(B)\right),$$
with $\alpha$ defined as in \eqref{eq_defalpha} and thus the algebras associated to $\star_H$, $\cdot_{T_H}$ and  $\star_{H'}$, $\cdot_{T_{H'}}$ are isomorphic via $$\alpha_{w}:\A_0\to \A^\prime_0\,,$$ where we recall that $\A_0$ is algebra $\star_H$--generated by $\cdot_{T_H}$--products of local functionals.

Hence, one may choose a suitable $H_+$ according to ones needs. However, since $\alpha_{d}(A)\neq A$ for functionals containing multiple field powers, statements like ``the potential is $\phi^4$'' are ambiguous in pAQFT, and in fact also in the standard treatment of QFT. They become non--ambiguous only if one says ``the potential is $\wick{\phi^4}_H$, i.e.  $\phi^4$ Wick--ordered with respect to $H_+$''. In pAQFT the corresponding non--ambiguous statement would be ``the potential is the functional $\phi^4$ in the algebra $\A_0$ constructed by means of $H_+$''. If one then passes to the algebra $\A^\prime_0$ constructed by means of $H^\prime_+$, the potential picks up quadratic and c--number terms as we shall compute explicitly below. Alternatively, this ambiguity may be seen to correspond to the renormalisation ambiguity of tadpoles in Feynman diagrams.

Given a Gaussian and Hadamard free field state $\Omega$, a convenient choice or representation of the algebra is to take $H_+=\Delta_+$, where $\Delta_+(x,y)=\langle \phi(x)\star_\Delta\phi(y)\rangle_\Omega\doteq\langle \phi(x)\phi(y)\rangle_\Omega$ is the two-point function of the free field in the state $\Omega$. This corresponds to standard normal--ordering and consequently in this  representation the expectation values of all expressions which contain non-trivial powers of the field vanish, i.e.
\begin{equation}\label{eq_expval} \langle A\rangle_\Omega = A|_{\phi=0}\doteq \langle \wick{A}_\Delta\rangle_\Omega\,.\end{equation}
Keeping the state $\Omega$ fixed, but passing on to a representation of the algebra with arbitrary $H_+$, the expectation value is computed as
$$ \langle A\rangle_\Omega = \alpha_w(A)|_{\phi=0}\doteq \langle \wick{A}_H\rangle_\Omega\,, \qquad w=\Delta_+-H_+\,,$$
for instance
$$\langle \phi^2(x)\rangle_\Omega = \alpha_w(\phi^2(x))|_{\phi=0}=\left(\phi^2(x)+w(x,x)\right)|_{\phi=0}=w(x,x)\doteq \langle \wick{\phi^2(x)}_H\rangle_\Omega \,,$$
which in more standard terms would be computed as
$$\label{eq_expval2}\langle \wick{\phi^2(x)}_H\rangle_\Omega =\lim_{x\to y}\left\langle \phi(x)\phi(y)-H_+(x,y)\right\rangle_\Omega=\lim_{x\to y}\left(\Delta_+(x,y)-H_+(x,y)\right)=w(x,x)\,.$$

In QFT in curved spacetimes normal--ordering is in principle problematic, because (pointlike) observables should be defined in a local and generally covariant way, i.e. they should only depend on the spacetime in an arbitrarily small neighbourhood of the observable localisation \cite{Brunetti:2001dx, Hollands:2001nf}. This is not satisfied for e.g. field polynomials Wick--ordered with $\Delta_+(x,y)$, because this distribution satisfies the Klein-Gordon equation and thus it encodes non--local information on the curved spacetime \cite{Hollands:2001nf}. It is still possible to compute in the convenient normal--ordered representation in the following way. In the example of $\phi^4$--theory, one defines the potential $\frac{\lambda}{4} \phi(x)^4$ as a local and covariant observable by identifying it with the corresponding monomial in a representation of the algebra furnished by a purely geometric $H_+$, i.e. a $H_+$ of the form \eqref{eq:hadamard} with $w=0$.

In other words, we set once and for all in the $H_+$--representation 
$$
V_H=\int_\M d^4x \sqrt{-g} \; \frac{\lambda}{4} \phi(x)^4\doteq \int_\M d^4x \sqrt{-g} \; \frac{\lambda}{4} \wick{\phi(x)^4}_H.
$$
This does not fix $V$ uniquely, because $H$ depends on the scale $M$ inside of the logarithm, but the freedom in defining $V_H$, and analogously the free/quadratic part of Klein--Gordon action, as above corresponds to the usual freedom in choosing the ``bare mass'' $m$, ``bare coupling to the scalar curvature'' $\xi$, ``bare cosmological constant'' $\Lambda$, ``bare Newton constant'' $G$, as well as the ``bare coefficients'' $\beta_1$, $\beta_2$ of higher--derivative gravitational terms in the extended Einstein--Hilbert--Klein--Gordon action
$$
{\cal S}(\phi,g_{ab})=\int_\M d^4x \sqrt{-g}\left(\frac{R-2\Lambda}{16 \pi G}+\beta_1 R^2 + \beta_2 R_{ab}R^{ab}-\frac{(\nabla \phi^2)}{2}-\frac{(m^2+\xi R )\phi^2}{2}-\frac{\lambda}{4}\phi^4\right).
$$
In order to switch to the normal-ordered representation, we use the map $\alpha_{w}$ defined in \eqref{eq_defalpha} where $w=\Delta_+-H_+$ is the state--dependent part of the Hadamard distribution $\Delta_+$ whose dependence on the choice of $M$ in $H_+$ corresponds to the above--mentioned freedom in the definition of the Wick--ordered Klein--Gordon action. That is, we have in the normal--ordered representation in the state $\Omega$
\begin{align}\label{eq:potentialcorrections}
V:=V_\Delta=\alpha_w(V_H)&=\int_\M d^4x \sqrt{-g} \; \frac{\lambda}{4} \phi(x)^4 + \frac{3\lambda}{2}w(x,x)\phi(x)^2+\frac{3\lambda}{4} w(x,x)^2\\
&\doteq\int_\M d^4x \sqrt{-g} \; \frac{\lambda}{4} \wick{\phi(x)^4}_\Delta + \frac{3\lambda}{2}w(x,x)\wick{\phi(x)^2}_\Delta+\frac{3\lambda}{4} w(x,x)^2\notag
\end{align}
We observe that the combination of the requirements that the interaction potential is a local and covariant observable and that, in order to compute expectation values in the state $\Omega$, one would like to compute in the convenient normal--ordered representation with respect to $\Omega$, leads to the introduction of an effective spacetime--dependent and state--dependent (squared) mass term $\mu(x)=3\lambda w(x,x)$ in the interaction potential which of course leads to additional Feynman graphs in perturbation theory, cf. Figures \ref{fig_propagators} and \ref{fig_2pf}. The field--independent term $\frac{3\lambda}{4} w(x,x)^2$ plays no role for computations of quantities which do not involve functional derivatives of the extended Einstein--Hilbert--Klein--Gordon action with respect to the metric (an example where it does play a role is the stress--energy tensor), just as the modification of the free action by the change of representation plays no role for the computation of such quantities. A similar phenomenon as in \eqref{eq:potentialcorrections} occurs in thermal quantum field theory on Minkowski spacetime, where the effective mass generated by changing from the normal--ordered picture with respect to the free vacuum state to the normal--ordered picture with respect to the free thermal state is termed ``thermal mass'', cf. \cite[Section 2.3.2.]{Lindner:2013ila} for details.

After these general considerations, we can proceed to compute as an example the two-point function of the interacting field $\phi_I$ in $\phi^4$ up to second order in $\lambda$, whereby $\phi_I$ is assumed to be in a state induced by a Gaussian Hadamard state of the free field. To this avail, we shall exclusively compute in the associated normal--ordered representation and thus omit the subscripts on the star product, and the time-ordered product, $\star:=\star_\Delta$, $\cdot_T\;:=\cdot_{T_\Delta}$.

We start from the Bogoliubov formula \eqref{eq_bogoliubov} and compute (from now on $\hbar=1$)
$$
S(V)=1+iV-\frac12 V\cdot_T V + O(\lambda^3)
$$
$$
S(V)^{\star -1}=1-iV+\frac12 V\cdot_T V-V\star V + O(\lambda^3)
$$
$$
\phi_I=\phi-i V\star \phi+i V\cdot_T\phi+\frac12\left( V\cdot_T V\right)\star \phi-V\star V\star\phi-\frac12 V\cdot_T V\cdot_T \phi+V\star(V\cdot_T\phi)+O(\lambda^3)\,.
$$
It remains to compute the $\star$-product of $\phi_I(x)$ and $\phi_I(y)$ and to set $\phi=0$ in the remaining expression in order to obtain the expectation value in the state $\Omega$. The result can as always be conveniently expressed in terms of Feynman diagrams, where we use the Feynman rules depicted in Figure \ref{fig_propagators}.

\begin{figure}[!htb]\begin{center}
\includegraphics[width=9.5cm]{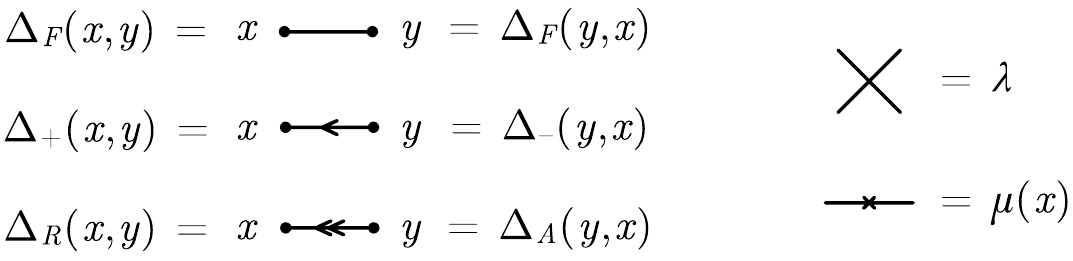}
\end{center}
\caption{\label{fig_propagators}The various propagators and vertices in $\phi^4$--theory, where $\mu(x)=3\lambda w(x,x)$.}
\end{figure}

 In the computation of $\langle\phi_I(x)\phi_I(y)\rangle_\Omega$, many expressions can be shortened considerably by using the relation $\Delta_F-\Delta_+=i\Delta_A$, in particular this holds for the external legs of the appearing Feynman diagrams. The resulting Feynman diagrams are depicted in Figure \ref{fig_2pf}.

\begin{figure}[!htb]\begin{center}
\includegraphics[width=11cm]{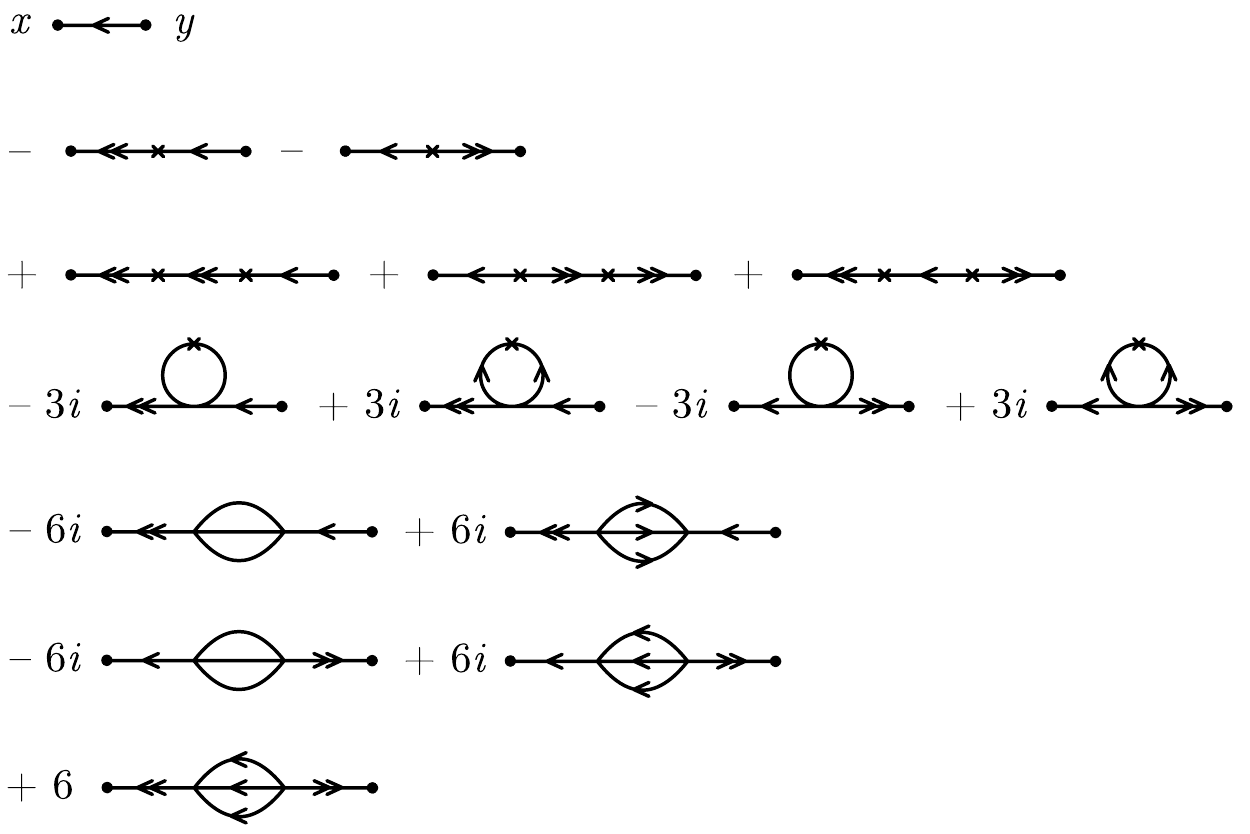}
\end{center}
\caption{\label{fig_2pf}The up--to--second--order contributions to the two--point (Wightman) function $\langle\phi_I(x)\phi_I(y)\rangle_\Omega$ of the interacting field with potential $\frac{\lambda}{4}\phi(x)^4+\frac{\mu(x)}{2}\phi(x)^2$. We omit the labels of the external vertices after the first line using the convention that the left external vertex is always the $x$-vertex.}
\end{figure}

\section{Analytic regularisation and minimal subtraction on curved spacetimes}

As discussed above, the main problem in using the Bogoliubov formula \eqref{eq:bogoliubov}
\[
\mathcal{R}_V(F) = \left. \frac{\hbar}{i}\frac{d}{d\lambda} S(V)^{-1}\star_H S(V+\lambda F) \right|_{\lambda = 0}
\]
for constructing interacting fields perturbatively is that it is given in terms of the $S$--matrix, which is the time--ordered exponential \eqref{def:Smatrix}. 
Unfortunately, the time--ordered product defined in terms of a ``deformation'' \eqref{def:TH} written by means of a Feynman propagator $H_F$ is well defined only on regular functionals because  
the singularities present in $H_F$ forbid their application to more general functionals.

In order to proceed there is the need of employing a renormalisation procedure to construct the time--ordered products.
In this work we discuss the use of certain analytic methods to solve this problem.
The procedure we shall pursue is the following. 
We deform the Feynman propagator by means of complex parameter $\alpha$ with values in the neighbourhood of the origin obtaining a function with distributional values $\alpha \mapsto H_F^{(\alpha)}$. The deformation we are looking for needs to be such that in the limit $\alpha \to 0$ we recover the ordinary Feynman propagator. Furthermore, when $\alpha$ is non--vanishing, but sufficiently small, pointwise powers of $H^{(\alpha)}_F$ and integral kernels of more complicated loop diagrams should be well--defined. If this is the case, since the corresponding distributions obtained in the limit $\alpha\to0$ are well defined outside of the total diagonal, the poles of $\alpha \mapsto H_F^{(\alpha)}$ and more complicated loop expressions are supported on the total diagonal. 
The idea, similar to what happens in dimensional regularisation, is that it is possible to renormalise these distributions by simply removing the poles.

\subsection{Analytic regularistion of time--ordered products and the minimal subtraction scheme}
\label{sec:analytic_general}

In order to discuss the analytic regularisation of time--ordered products, we employ the notation used e.g. in \cite{dfkr} which efficiently encodes the full combinatorics of Feynman diagrams in a compact form. Namely, the time--ordered product of $n$ local functionals $V_1, \dots, V_n$ can be formally defined in the following way\footnote{\label{foot:T1}In fact, in view of locality and covariance a better definition of the time--ordered product is $\mathcal{T}_1(V_1)\cdot_{T_H}  \dots \cdot_{T_H}    \mathcal{T}_1(V_n)    
:=
\mathcal{T}_n(V_1\otimes \dots \otimes V_n) $ where $\mathcal{T}_1:\F_\loc\to\F_\loc\subset\A_0$ plays the role of identifying local and covariant (smeared) Wick polynomials as particular elements of the free algebra $\A_0$, cf. \cite{Hollands:2004yh}. As we shall not touch upon this point in our renormalisation scheme, we choose to omit $\mathcal{T}_1$ in our formulas for simplicity.}
\begin{equation}\label{eq:time-ordered-product}
V_1\cdot_{T_H}  \dots \cdot_{T_H}    V_n    
:=
\mathcal{T}_n(V_1\otimes \dots \otimes V_n) 
:=  m \circ T_n (V_1\otimes \dots \otimes V_n)\,,
\end{equation}
where $m$ denotes the pointwise product $m(F_1\otimes \dots \otimes F_n)(\phi) = F_1(\phi)\dots F_n(\phi)$
and the operator $T_n$ is written in terms of an exponential 
\beq\label{def:exponentialT}
T_n = \exp\left(\sum_{1\leq i<j \leq n}\Delta_{ij}\right)= \prod_{1\leq i<j\leq n} \sum_{l_{ij} \geq 0}^\infty  \frac{\Delta_{ij}^{l_{ij}}}{l_{ij}!}
\eeq
with
\beq\label{def:exponentialT2}
\Delta_{ij} :=   \left\langle H_F, \frac{\delta^2}{\delta \phi_i \delta \phi_j} \right\rangle.
\eeq
Here the functional derivative $\frac{\delta}{\delta \phi_i}$ acts on the $i-$th element of the tensor product 
$V_1\otimes \dots \otimes V_n$ and $H_F=H_++i\Delta_A$ is the time--ordered version of the Hadamard distribution $H_+$ entering the construction of the free algebra $\A_0$ via $\star_H$.
The exponential \eqref{def:exponentialT} admits the usual representation in terms of Feynman graphs. More precisely, it can be written as a sum over all graphs $\Gamma$ in $\mathcal{G}_n$, the set of all graphs with vertices $V(\Gamma)= \{ 1,\dots, n\}$ and $l_{ij}$ edges $e\in E(\Gamma)$  joining the vertices $i,j$. Furthermore, in this construction, there are no tadpoles $l_{ii}=0$ (cf. Section \ref{sec:relationPAQFT} for details on why these are absent) and the edges are not oriented $l_{ij}=l_{ji}$. With this in mind 
\begin{equation}\label{eq:tau-gamma}
T_n = \sum_{\Gamma\in \mathcal{G}_n}  \frac{1}{N(\Gamma)}    \left\langle  \tau_\Gamma  , \frac{\delta^{2|E(\Gamma)|}}{  \prod_{i \in V(\Gamma)} \prod_{E(\Gamma)\ni e \supset i}    \delta \phi_i(x_{i}) } \right\rangle,
\end{equation}
 where $N(\Gamma)= \prod_{i<j} l_{ij}! $ is a numerical factor counting the possible permutations among the lines joining the same two vertices, the second product $\prod_{e \supset i} $ is over the edges having $i$ as a vertex and $x_{i}$ is a point in $\M$ corresponding to the vertex $i$.
Moreover, $\tau_\Gamma$  is a distribution which is well--defined outside of all partial diagonals, namely on
$\M^n\setminus D_n$, where 
\beq\label{def:alldiagonals}
D_n:=\{x_1,\ldots,x_n\,|\, x_i=x_j \text{ for at least one pair } (i,j),\, i\neq j \}\,
\eeq
and $\tau_\Gamma$ has the form
\begin{equation}\label{eq:tgamma}
\tau_\Gamma = \prod_{e=(i,j)\in E(\Gamma)} H_F(x_{i},x_{j})=\prod_{1\le i<j\le n} H_F(x_{i},x_{j})^{l_{ij}}.
\end{equation}
The a priori restricted domain of $\tau_\Gamma$ is the reason why $T_n$ defined as above is not a well--defined operation on $\mathcal{F}_\loc^{\otimes n}$. In this context we recall that the total diagonal $d_n\subset D_n$ is defined as 
\beq
d_n := \{(x,\dots, x) , x\in \M\}\subset \M^n\,.
\eeq
In order to complete the construction we need to extend the obtained distributions to the diagonals $D_n$. This is not a straightforward limit because the singular structure of the Feynman propagator $H_F$ contains the one of the $\delta$--distribution and because pointwise products of the latter distribution are ill--defined. Consequently, a renormalisation procedure needs to be implemented in order to extend $\tau_\Gamma$ to the full $\M^n$. This extension is in general not unique, but subject to renormalisation freedom.

Here we shall discuss a procedure to extend the distributions $\tau_\Gamma$ to $D_n$ called {\bf minimal subtraction (MS)}, which makes use of an analytic regularisation $\Delta^{\alpha_{ij}}_{ij}$ of $\Delta_{ij}$ given in terms of a family of deformations $H_F^{\alpha_{ij}}$ of the Feynman propagator $H_F$ parametrised by complex parameters $\alpha_{ij}$ contained in some neighbourhood of $0\in\bbC$. To this end, we follow \cite{dfkr} and call $t^{(\alpha)}$ an analytic regularisation of a distribution $t$ defined outside of a point $x_0\in\M$ if for all $f\in\D(\M)$ $\langle t^{(\alpha)}, f\rangle$ is a meromorphic function in $\alpha$ for $\alpha$ in some neighbourhood of $0$ which is analytic for $\alpha\neq 0$. Moreover $t^{(\alpha)}$ may be extended to $x_0$ for $\alpha\neq 0$ whereas $\lim_{\alpha\to 0}t^{(\alpha)}=t$ on $\M\setminus\{x_0\}$.

We shall introduce an analytic regularisation of the Feynman propagator $H_F$ in the following section, but the basic idea of the MS--scheme is independent of the details of the analytic regularisation. Namely, given 
any analytic regularisation $H^{(\alpha)}_F$ of $H_F$, we repeat the formal construction of $T_n$ presented above by replacing $H_F$ by $H^{(\alpha)}_F$ in \eqref{def:exponentialT2} and $\Delta_{ij}$ by the induced $\Delta^{\alpha_{ij}}_{ij}$ in \eqref{def:exponentialT}. Proceeding in this way we define 
\[
T^{(\balpha)}_n := e^{\sum_{i<j}\Delta^{\alpha_{ij}}_{ij}} \qquad \text{with}\qquad\balpha := \{\alpha_{ij}\}_{i<j}\,,
\]
and the corresponding integral kernels $\tau^{(\balpha)}_\Gamma$ of Feynman graphs $\Gamma$ in analogy to \eqref{eq:tau-gamma}. We expect that the distributions $\tau^{(\balpha)}_\Gamma$ are multivariate meromorphic functions which have poles at the origin for some of the $\alpha_{ij}$. 
Hence, in order to obtain well--defined distributions in the limit $\alpha_{ij}$ to $0$ and consequently a renormalised time--ordered product $\cdot_{T_H}$, all these poles need to be subtracted.

The properties of the analytically regularised Feynman propagator imply that $\tau^{(\balpha)}_\Gamma$ is well--defined on $ \M^n \setminus D_n $ \eqref{def:alldiagonals} even if all $\alpha_{ij}$ are vanishing. Since $\tau^{(\balpha)}_\Gamma$ is a multivariate meromorphic function in $\balpha$  which is analytic if restricted to $\M^n\setminus D_n$, we may deduce that the principal part of $\tau^{(\balpha)}_\Gamma$ for some $\alpha_{ij}$ must be supported on a partial diagonal of $\M^n$. In fact, in order for the time--ordered products to fulfil the factorisation property \eqref{eq:causal-factorisation}, the subtraction of the principal parts of $\tau^{(\balpha)}_\Gamma$ needs to be done in such a way that at each step only local terms are subtracted. However, the previous discussion only implies that the support of the principal parts is contained in $D_n$, i.e. the union of all the partial diagonals in $\M^n$. In order to satisfy the causal factorisation property, the principal parts need to be removed in a recursive way starting from the partial diagonals corresponding to two vertices and proceeding with the partial diagonals corresponding to an increasing number $m\le n$ of vertices $\mathfrak{d}_{I}:=\{ (x_1,\dots, x_n) \in M^n, x_i=x_j, i,j \in I\subset \{1,\dots, n\} , |I|=m\}$.

The correct recursion procedure is implemented by the so called Epstein--Glaser forest formula, which is a position--space analogue of the Zimmermann forest formula, see  \cite{Hollands:2010pr, Keller, dfkr} for a careful analysis of the subject. We shall here follow the treatment discussed in \cite{dfkr}. To this end, we consider the set of indices $\overline{n} := \{1,\dots , n\}$ and define a forest $F$ as 
\[
F =\{ I_1,\dots, I_k\}, \qquad I_j \subset \overline{n}\qquad\text{and} \qquad |I_j|\geq 2 \,,
\]
where for every pair $I_i,I_j\in F$
\[
I_i\cap I_j = \emptyset  \qquad \text{or}   \qquad I_i \subset I_j   \qquad \text{or}   \qquad  I_j\subset I_j.
\]
The set of all forests of $n$ indices together with the empty forest $\{\}$ is indicated by $\mathfrak{F}_{\overline{n}}$.

For every subset $I\subset \overline{n}$ we indicate by $R_I$ the operator which extracts the principal part with respect to $\alpha_I$ of a multivariate meromorphic function $f(\{\alpha_{ij}\}_{i<j})$, where for every $i,j \in I$, $\alpha_{ij}=\alpha_I$, and multiplies it with $-1$:
\begin{equation}\label{eq:projection}
R_I f := - \pp  \lim_{\scriptstyle\alpha_{ij} \to \alpha_I\atop\forall i,j\in I} f(\{\alpha_{ij}\}_{i<j}).
\end{equation}
We complement this definition by setting $R_{\{\}}$ to be the identity.

Given all these data, we define the renormalised time--ordered product in the MS--scheme as in e.g. \cite[Theorem 3.1]{dfkr} by
\begin{equation}\label{eq:forset-formula}
\mathcal{T}_n = \left(\mathcal{T}_n\right)_\ms := \lim_{\balpha \to 0}   m \circ   \left( \sum_{F\in\mathfrak{F}_{\overline{n}}} \prod_{I\in F} R_I\right)  \circ    T^{(\balpha)}_n,
\end{equation}
where, in the product over $I\in F$, $R_I$ appears before $R_J$ if $I\subset J$. 
Furthermore, for each graph $\Gamma$, the limit $\balpha=\{\alpha_{ij}\}_{i<j}\to 0$ is computed by setting $\alpha_{ij}= \alpha_\Gamma$ for every $i<j$ before taking the sum over the forests and finally considering the limit $\alpha_\Gamma$ to $0$. In this context we recall that, for every element of the sum over $\mathfrak{F}_{\overline{n}}$, part of the limit $\alpha_{ij}\to \alpha_\Gamma$ is already taken by applying $R_I$, see \eqref{eq:projection}. 

Given the renormalised $\mathcal{T}_n$ in the MS--scheme, the corresponding local $S$--matrix may be constructed as 
\[
S(V) = \sum^\infty_{n=0}\frac{i^n}{\hbar^n n!}\mathcal{T}_n(V\otimes  \dots \otimes V)
\]
for any local interaction Lagrangean $V$.

In order to implement the minimal subtraction scheme as outlined above we first need to specify an analytic regularisation $H^{(\alpha)}_F$ of the Feynman propagator $H_F$ on generic curved spacetimes. Afterwards we have to demonstrate that for all graphs $\Gamma\in\G_n$ the analytically regularised integral kernels 
\beq\label{def:regularisedamplitudes}
\tau^{(\balpha)}_\Gamma = \prod_{e=(i,j)\in \Gamma} H^{\alpha_{ij}}_F(x_{i},x_{j})=\prod_{1\le i<j\le n} \left(H^{\alpha_{ij}}_F(x_{i},x_{j})\right)^{l_{ij}}.
\eeq
appearing in
\begin{equation}\label{eq:tau-gamma_reg}
 T^{(\balpha)}_n = \sum_{\Gamma\in \mathcal{G}_n}  \frac{1}{N(\Gamma)}    \left\langle  \tau^{(\balpha)}_\Gamma  , \frac{\delta^{2|E(\Gamma)|}}{   \prod_{i \in V(\Gamma)} \prod_{e \supset i}    \delta \phi_i(x_{i}) } \right\rangle 
\end{equation}
satisfy the properties necessary for the implementation of the MS--scheme. In particular we need to demonstrate that the distribution $\tau^{(\balpha)}_\Gamma$, which is a priori defined only on $\M^n\setminus D_n$, can be uniquely extended to the full $\M^n$ without renormalisation, where the uniqueness of this extension is important in order to obtain a definite renormalisation scheme. Moreover, we need to show that this distribution $\tau^{(\balpha)}_\Gamma\in\D^\prime(\M^n)$ is weakly meromorphic in $\balpha$ in a neighbourhood of 0, where in view of the forest formula it is only necessary to show that, setting $\alpha_{ij}=\alpha_I$ for all $i,j\in I$, $\tau^{(\balpha)}_\Gamma$ is weakly meromorphic in $\alpha_I$. Additionally, we need to prove that, if $\tau_\Gamma$ prior to regularisation is well--defined outside of the partial diagonal $d_{I}$, then the pole of $\tau^{(\balpha)}_\Gamma$ with $\alpha_{ij}=\alpha_I$ for all $i,j\in I$ in $\alpha_I$ is supported on $d_I$ and thus local. Finally, we need to prove that our MS--scheme satisfies all properties given in \cite{Hollands:2001b,Hollands:2004yh} which a physically meaningful renormalisation scheme on curved spacetimes should satisfy, and we need to provide means to explicitly compute the minimal subtraction, which after all is the main motivation for this work.

Our plan to construct the mentioned quantities and to prove their required properties is as follows.

\begin{enumerate}
\item In Section \ref{sec:analyticFeynman} we construct an analytic regularisation $H^{(\alpha)}_F$ of the Feynman propagator based on the observation that locally $H_F$ is of the form \eqref{eq:hadamard} up to considering instead of $\sigma_+$ the half squared geodesic with the Feynman $\epsilon$--prescription $\sigma_F := \sigma + i\epsilon$. Motivated by the fact that the singular structure of $H_F$ originates from the form in which $\sigma_F$ appears, we set locally
\begin{equation}\label{eq:anal-feynman}
H^{(\alpha)}_F := \lim_{\epsilon\to 0^+} \frac{1}{8\pi^2}\left(\frac{u}{M^{2\alpha}\sigma_F^{1+\alpha}} + \frac{v}{\alpha} \left(1-\frac{1}{M^{2\alpha}\sigma_F^{\alpha}}\right)\right)+w,
\end{equation}
where we use the (arbitrary but fixed) mass scale $M$ present in \eqref{eq:hadamard} also for preserving the mass dimension of $H_F$ in the regularisation.
\item In Proposition \ref{pr:prod-sigma} we then prove that the relevant distributions
\beq\label{def:amplitudesigma}t_\Gamma^{(\balpha)} := \prod_{1\le i<j\le n} \frac{1}{\sigma_F^{l_{ij}(1+\alpha_{ij})}}\in \D^\prime(\M^n\setminus D_n)\eeq
are multivariate analytic functions. The distribution \eqref{def:amplitudesigma} only displays the most singular contribution of $\tau^{(\balpha)}_\Gamma$ \eqref{def:regularisedamplitudes}, but the subleading contributions are clearly of the same form up to replacing some of the factors $(1+\alpha_{ij})$ in the exponents by $\alpha_{ij}$ or 0.
\item In order to show that $t_\Gamma^{(\balpha)}$ can be uniquely extended from $\M^n\setminus D_n$ to $\M^n$ in a weakly meromorphic fashion, i.e. that the singularities relevant for the forest formula are poles of finite order, we follow a strategy similar to the one used in \cite{Hollands:2001b} and consider a scaling expansion with respect to a suitable scaling transformation. We first argue in Proposition \ref{pr:regularisation} that an analytically regularised distribution $t^{(\alpha)}\in\D^\prime(\M^n\setminus d_n)$, which can be written as a sum of homogeneous terms with respect to this scaling transformation plus a sufficiently regular remainder, can be extended to $\M^n$ in a weakly meromorphic way, were the uniqueness of the extension follows from its weak meromorphicity. In Proposition \ref{pr:set}, we give a sufficient condition for the existence of such a homogeneous expansion and we demonstrate in Proposition \ref{pr:almost-homo} that the distributions $t_\Gamma^{(\balpha)}$ satisfy this condition.

\item The above--mentioned results are proved by means of generalised Euler operators (see \cite{DangPHD} for a related concept) which can be written abstractly in terms of a scaling transformation, but also in terms of covariant differential operators whose explicit form can be straightforwardly computed as we argue in Section \ref{sec:differentialEuler}.  In Proposition \ref{pr:expose-poles} we use these operators in order to demonstrate how the full relevant pole structure of $t_\Gamma^{(\balpha)}$ can be computed, thus showing the practical feasibility of the MS--scheme. We find that our renormalisation scheme corresponds in fact to a particular form of differential renormalisation and expand on this by computing a few examples in Section \ref{sec_fishsunset}.

\item Finally, in Proposition \ref{pr:propertiesscheme} we prove that the MS--scheme satisfies the axioms of \cite{Hollands:2001b,Hollands:2004yh} for time--ordered products and in addition preserves invariance under any spacetime isometries present.
\end{enumerate}

\begin{rem}\label{rem:geodesicneighbourhood}
The local Hadamard expansion  \eqref{eq:hadamard} of $H_F$ and correspondingly the analytically continued $H^{(\alpha)}_F$ defined in \eqref{eq:anal-feynman} are only meaningful on normal neighbourhoods $\N$ of $(\M,g)$. In order to define $H^{(\alpha)}_F$ and the induced distributions $\tau^{(\balpha)}_\Gamma$\eqref{def:regularisedamplitudes} globally, we may employ suitable partitions of unity. Rather than providing general and cumbersome formulas, we prefer to illustrate the idea at the example of the triangular graph
$$\tau_\Gamma=H_{F,13}H_{F,23}H_{F,12}^2.:=H_{F}(x_1,x_3)H_{F}(x_2,x_3)H_{F}(x_1,x_2)^2$$
the renormalisation of which is discussed in detail in Section \ref{sec:complicatedgraph}. 

We recall that a corollary of Lemma 10 of Chapter 5 in \cite{ONeill} guarantees that there exists a covering $\mathcal{C}$ of $\M$ consisting of open geodesically convex sets such that $\N_i\cap \N_j$ is geodesically convex for every $\N_i,\N_j \in \mathcal{C}$\footnote{We would like to thank Valter Moretti for pointing this result out to us.}. With this $\mathcal{C}$ at our disposal, we define the sets 
\[
\N_{12}:=\bigcup_{\N\in\mathcal{C}} \N\times \N\subset \M^2\,,\qquad \N_{123}:=\bigcup_{\N\in\mathcal{C}}\N\times \N\times \N\subset \M^3.
\]
We call sets of the form $\N_{12}$ and $\N_{123}$ a {\bf normal neighbourhood of the total diagonal}. This definition is essentially motivated by the fact that for every $x\in \M$ we can find a normal neighborhood $\N_x\in \mathcal{C}$ of $x$ in $\M$. The squared geodesic distance $\sigma$ is then well defined on $\mathcal{N}_{12}$, whereas the same is in general not true if we replace $\mathcal{C}$ in the previous formula with a covering of $\mathcal{M}$ formed by all open geodesically convex sets.

Setting $\sigma_{ij}:=\sigma(x_i,x_j)$, we observe that $\sigma_{12}$ is well--defined on $\N_{12}$, and that $\sigma_{12}$ , $\sigma_{13}$ and $\sigma_{23}$ are well--defined on $\N_{123}$. We now consider smooth and compactly supported functions $\chi_{12}\in\D(\N_{12})$, $\chi_{123}\in\D(\N_{123})$ which are such that $\chi_{12}=1$ on $d_2\subset \N_{12}$ and $\chi_{123}=1$ on $d_3\subset \N_{123}$. Note that by construction $\chi_{12}$ and $\chi_{123}$ vanish outside of $\N_{12}$ and $\N_{123}$ respectively. We may now define the analytically regularised distribution $\tau^{(\balpha)}_\Gamma$ by setting
\begin{align*}\tau^{(\balpha)}_\Gamma:=&\,H^{(\alpha_{13})}_{F,13}H^{(\alpha_{23})}_{F,23}\left(H^{(\alpha_{12})}_{F,12}\right)^2\chi_{12}\chi_{123} + H_{F,13}H_{F,23}H_{F,12}^2(1-\chi_{12})\\&+H_{F,13}H_{F,23}\left(H^{(\alpha_{12})}_{F,12}\right)^2\chi_{12}(1-\chi_{123})\,,
\end{align*}
where the Feynman propagators are regularised as in \eqref{eq:anal-feynman}. By construction, $\tau^{(\balpha)}_\Gamma$ is globally well--defined and the analysis outlined above and performed in the following sections implies that it can be uniquely extended to a weakly meromorphic distribution on the full $\M^3$. Moreover, the local pole contributions corresponding to $\alpha_{12}=\alpha_I$ with $I=\{1,2\}$ and $\alpha_{12}=\alpha_{13}=\alpha_{23}=\alpha_J$ with $J=\{1,2,3\}$ are clearly independent of the choice of $\chi_{12}$, $\chi_{123}$ and $\N_{12}$, $\N_{123}$ such that the MS--regularised amplitude $(\tau_\Gamma)_\ms$ is both globally well--defined and independent of the quantities entering the global definition of the analytic regularisation.

Keeping this approach to define global analytically regularised quantities in mind, we shall for simplicity work only with local quantities in the following.
\end{rem}


\subsection{Analytic regularisation of the Feynman propagator \texorpdfstring{$H_F$}{HF} on curved spacetimes}
\label{sec:analyticFeynman}

Following the plan outlined in Section \ref{sec:analytic_general}, we would like to define an analytic regularisation $H^{(\alpha)}_F$ of $H_F$ by \eqref{eq:anal-feynman}. To this end, we start our analysis by constructing the distribution $1/\sigma_F^{1+\alpha}$ in $\M^2$ for $\alpha\in \mathbb{C} \setminus \mathbb{N}$. As anticipated in Section \ref{sec:analytic_general} we shall make use of scaling properties of $1/\sigma_F^{1+\alpha}$ and the induced quantities $t^{(\balpha)}_\Gamma$ \eqref{def:amplitudesigma} with respect to a particular geometric scaling transformation. 

For every pair of points $x_1,x_i$ in a normal neighbourhood $\N\subset (\M,g)$ there exists a unique geodesic $\gamma$ connecting $x_1$ and $x_i$. We shall assume that $\gamma:\lambda \mapsto x_i(\lambda)$ is affinely parametrised and that $x_i(0) =x_1$ whereas $x_i(1) = x_i$. 
For all $\lambda\ge0$ and all $f\in \D(\N_n)$ with $\N_n\subset \M^n$ a normal neighbourhood of the total diagonal $d_n$ (cf. Remark \ref{rem:geodesicneighbourhood}), the geometric scaling transformation we shall consider is
\begin{equation}\label{eq:n-dim-scaling}
f_\lambda := \lambda^{4(n-1)}f(x_1,x_2(\lambda ), \dots, x_n(\lambda))  \prod_{i=2}^n\frac{\sqrt{g(x_i(\lambda ))}}{\sqrt{g(x_i)}}\,,
\end{equation}
where $g(x)$ is the absolute value of the determinant of the metric expressed in normal coordinates. For $\lambda>1$ it may happen that $x_i(\lambda)$ lies outside of $\N_n$ and is thus not well--defined in general. In this case we set $f_\lambda=0$ which is well--defined because $f=0$ outside of $\N_n$. For later purposes, we recall that the determinant of the metric computed in normal coordinates centred at $x_1$ is such that 
\[
\sqrt{g(x_i)} = \frac{1}{u^2(x_1,x_i)}\,,
\]
where $u$ is the Hadamard coefficient in \eqref{eq:hadamard} and $u^2$ is the van Vleck--Morette determinant, see e.g.  \cite[(8.5)]{Poisson:2011nh}.

By means of this transformation, relevant information about the behaviour of a distribution in the neighbourhood of the total diagonal $d_n$ can be obtained. 

We recall two definitions and a few results which we shall use in the following. To this end we consider a distribution $t$ in $\D^\prime(\M^{n}\setminus d_n)$. The restrictions of $t$ to the sets 
\[
\mathcal{O}_x = \left\{   (x,y) \in \M^n \,\big|\, y \in \M^{n-1} \right\}
\]
for $x\in\M$ are well defined if $\WF(t)$ is disjoint from the {\bf conormal bundle}
 $N^*\mathcal{O}_x=\{(y,k) \in T^*\M \,\big|\,   \langle k, \xi \rangle = 0 \, ,\forall \xi\in T_y\mathcal{O}_x  \}$, see e.g. Theorem 8.2.4 in the book of H\"ormander \cite{Hormander}. By assumption, these restrictions are defined everywhere on $\mathcal{O}_x$ up to the single points 
$\mathcal{O}_x \cap d_n$. The extension of the restricted distributions to that point can be constructed by means of particular scaling properties of that distribution, see e.g. Theorem 5.2 and Theorem 5.3 in \cite{Brunetti-Fredenhagen:2000}. However, in order to discuss the extension of $t$ to the full diagonal $d_n$ it is useful to introduce the scaling degree of $t$ towards $d_n$ in the following way, which is an adaption of the notion of a transversal scaling degree presented in  \cite[Section 6]{Brunetti-Fredenhagen:2000} to our setting.

\begin{defi}\label{def:scalingdegree}
Let $\Gamma \subset T^*\M^n$ be a closed conical set which is equal to the conormal bundle $N^*d_n$ on $d_n$ and disjoint from $N^*\mathcal{O}_x$ for every $x$, and let $t$ be a distribution in $\D_\Gamma^\prime(X)$, $X\in\{\M^n,\M^n\setminus d_n\}$ i.e. $t\in\D^\prime(X)$ and $\WF(t)\subset \Gamma$. The {\bf scaling degree} of $t$ towards $d_n$ 
is defined as 
\[
\sd_{d_n}(t) := \inf\left \{w \in \mathbb{R}\,\big|\, \lim_{\lambda \to 0^+}\lambda^{w}  t_\lambda = 0\right\}
\] 
where the distributions $t_\lambda(f) := t(f_{1/\lambda})$ are defined in terms of the scaling transformation \eqref{eq:n-dim-scaling}, and where the limits are taken in the topology of $\D^\prime_\Gamma(\M^n)$. 
\end{defi}

If a distribution has scaling degree towards $d_n$ lower than the total dimension of the scaled coordinates $4(n-1)$, then it possesses a unique extension towards $d_n$ with the same scaling degree, see e.g. Theorem 6.9 of \cite{Brunetti-Fredenhagen:2000}. The proof of existence of the extensions presented there is constructive and can be obtained as the limit to large $n$ of $t (1-\theta_n)$, where $\theta_n$ are suitable smoothed characteristic functions of a neighborhood of the diagonal whose supports become smaller and smaller for large $n$. In particular, the previous notion of scaling degree permits to control the wave front set of every $t (1-\theta_n)$ and of the limiting distribution as well. In fact, the extended distribution can be restricted to every $\mathcal{O}_x$ as well because $\Gamma$ and $N^*\mathcal{O}_x$ are disjoint sets. 
The uniqueness of the extension descends from the uniqueness of the extension of the distributions restricted to $\mathcal{O}_x$ and from the observation that the scaling degree of the restricted distributions towards the point $(x,\dots,x) \in \mathcal{O}_x$ is such that $\sd(\left. t\right|_{\mathcal{O}_x}) \leq \sd_{d_n}(t)$. If not strictly necessary we shall drop the suffix $d_n$ in the following. Having this result at disposal implies that, in order to construct an extension towards $d_n$, it is sufficient to construct extensions of the distribution restricted to $\mathcal{O}_x$. 

\begin{rem} The distributions relevent for our analysis are the ones of the form \eqref{def:regularisedamplitudes}, respectively \eqref{def:amplitudesigma}. As we discuss briefly in the proof of Proposition \ref{pr:prod-sigma}, these distributions satisfy the micro local spectrum condition introduced in \cite{bfk:1996} and, consequently, also the wave front set condition required in Definition \ref{def:scalingdegree}.
\end{rem}

The scaling degree towards a partial diagonal may be defined in analogy to Definition \ref{def:scalingdegree}. 
The same geometric transformation \eqref{eq:n-dim-scaling} can be used to introduce relevant homogeneity properties of a distribution.

\begin{defi}\label{def:homogeneous}
A distribution $t\in \D^\prime(\M^n)$ or $t\in \D^\prime(\M^n\setminus d_n)$, which satisfies the equality 
\[
\lambda^{\delta} \langle t,f_\lambda \rangle    = \langle t,f \rangle \qquad\forall \lambda >0
\]
under transformations of the form \eqref{eq:n-dim-scaling} for all $f\in\D(\N_n\setminus d_n)$ and for a $\delta\in\mathbb{C}$, is called {\bf homogeneous of degree} $\delta$. 
\end{defi}

These definitions imply by direct inspection that a distribution which is homogeneous of degree $\delta$ and whose wave front set restricted to $d_n$ is contained in $N^*d_n$ 
has scaling degree $-\Ree(\delta)$. We further notice that homogeneous distributions $t\in\D(\M^n\setminus d_n)$ 
whose $\WF(t)\cap T^*_{d_n}\M$ is contained in $N^* d_n$, possess  unique extensions to $\M^n$ with the same degree of homogeneity $\delta$ if $-(\delta+4(n-1))\notin \mathbb{N}$.
A proof of this claim can be obtained following very closely the proof of Theorem 3.2.3 in \cite{Hormander}. 
In particular, the uniqueness of these extensions descends once again from the uniqueness of the extensions of the distributions restricted to $\M^n_x$, whereas the existence of the extension is obtained employing a fundamental solution of the Euler orperator $E_1$ we shall introduce below in \eqref{def:geneuler}.


In the next proposition we introduce the distributions we shall use as building blocks for the construction of regularised Feynman propagators on Lorentzian manifolds. Although not completely analogous, the construction we are going to present is similar  to the extension of Riesz distributions to curved spaces presented in Section 1.4 of \cite{BGP}. 
In particular, here we shall discuss the boundary value of $1/(\sigma+i\epsilon)^{\alpha}$ for $\epsilon\to0$ while ordinary Riesz distributions are 
related to the antisymmetric part of a different boundary value of the functions $1/\sigma^\alpha$.

\begin{propo}\label{pr:sigma-1}
Consider a normal neighbourhood $\N_2\subset\M^2$ of $d_2$ (cf. Remark \ref{rem:geodesicneighbourhood})
and the following expression for $\alpha\in \mathbb{C}$ and $f\in \D(\N_2)$ 
\[
\left\langle \frac{1}{\sigma^\alpha_F}, f \right\rangle := \lim_{\epsilon\to0^+ } \int_{\M^2} \frac{1}{(\sigma(x,y)+i\epsilon)^{\alpha}} f(x,y)  d\mu_g (x) d\mu_g (y)\,.
\]
Then the following statements hold.
\begin{enumerate}
\item $1/{\sigma^\alpha_F}$ restricted to $\D(\N_2\setminus d_2)$ is a distribution which is weakly analytic in $\alpha$.
\item $1/{\sigma^\alpha_F}$ is homogeneous of degree $-2\alpha$ with respect to transformations of the form \eqref{eq:n-dim-scaling} for $f\in \D(\N_2\setminus d_2)$.
\item $1/{\sigma^\alpha_F}$ is well--defined as a distribution on $\N_2$ for $2\alpha-4\notin \mathbb{N}$. 
Furthermore, for all $f\in\D(\N_2)$ $\langle 1/{\sigma^\alpha_F},f\rangle$  is analytic for $2\alpha-4\notin \bbN$ and meromorphic for $\alpha \in \bbC$ with simple poles at $2\alpha-4\in \mathbb{N}$. 
\end{enumerate}
\end{propo}
\begin{proof}
$a)$ For every $x \in \M$ we fix a normal coordinate system $\xi_x:y\to{\mathbb{R}^4}$ in order to parametrise points $y$ in a normal neighbourhood of $x$. Consequently, on $\N_2$ the squared geodesic distance divided by $2$ can be easily expressed as
\[
\sigma(x,y)= \frac{1}{2}\eta(\xi_x(y),\xi_x(y)) = \frac{1}{2}\xi_x^{a}{\xi_x}_{a}\,,
\]
where $\eta$ is the standard Minkowski metric given in Cartesian coordinates. Furthermore, 
\begin{equation}\label{eq:normal-integral}
\left\langle \frac{1}{\sigma^\alpha_F}, f \right\rangle = 
\lim_{\epsilon\to0^+ } \int_\M \int_{\mathbb{R}^4}  \frac{2^{\alpha}}{(\xi_x^a{\xi_{x}}_a+i\epsilon)^{\alpha}} f(x,\xi_x) \sqrt{g(\xi_x)}\; d^4\xi_x\; d\mu_g (x).
\end{equation}
which is well defined for $f\in\D(\N_2)$. 

Observe that $1/(\xi^{a}{\xi}_{a})^\alpha$  for $\xi^{a}\in \{z\in \mathbb{C}^4 \,|\,\Imm(z) \in V^\pm \}$, where $V^\pm$ is the forward or past light cone with respect to the Minkowski metric, is  analytic both in $\xi$ and $\alpha$.  Furthermore, in the limit $\epsilon\to0^+$,  $1/(\xi^{a}{\xi}_{a} + i \epsilon)^\alpha$ can be seen as the boundary value of that analytic function. Since this function grows at most polynomially for large $1/\Imm(\xi^a\xi_a)$ its boundary value defines a distribution, see e.g. \cite[Theorem 3.1.15]{Hormander}. The analytic dependence on $\alpha$ is weakly preserved in the limit $\epsilon\to0^+$, and thus the resulting distribution is weakly analytic.

$b)$  The transformation defined in \eqref{eq:n-dim-scaling} acts on points parametrised by normal coordinates as $\xi \to \lambda \xi $. Furthermore, $1/(\xi^{a}{\xi}_{a})^\alpha$ on $A \subset \mathbb{C}^4$ is homogenous of degree $2\alpha$ with respect to the transformation $\xi\to \lambda \xi$. The statement follows from this observation, taking into account \eqref{eq:n-dim-scaling} and \eqref{eq:normal-integral}.

$c)$ Theorem 3.2.3 in \cite{Hormander} ensures that the distribution $\dot{t}\doteq1/(\xi^{a}{\xi}_{a})^\alpha\in\D^\prime(\bbR^4\setminus  0)$ has a unique extension $t$ to $0$ preserving the degree of homogeneity for every $2\alpha-4 \notin \bbN$. 
Hence, $\mathbb{I}\otimes t$ defines a distribution on $\M \times \bbR^4$.
Finally, notice that there exists a neighbourhood $\OO$ of the diagonal in $\M_2$
where $f(x,\xi_x) \sqrt{g(\xi_x)}$ is a smooth and compactly supported function for every $f\in C^\infty_0(\OO)$. Consequently, the statement follows by choosing a partition of unity adapted to $\OO$. We refer to the discussion after Definition \ref{def:homogeneous} for general arguments why in fact Theorem 3.2.3 in \cite{Hormander} can be generalised to extensions of suitable homogeneous distributions to the diagonal rather than a point.
\end{proof}

The previous proposition guarantees that $1/\sigma_F^{\alpha}$ is weakly meromorphic in $\alpha$ with simple poles at $2\alpha-4\in\bbN$. This property is preserved under taking linear combinations and multiplication by smooth functions. Consequently, the analytically regularised Feynman propagator $H^{(\alpha)}_F$ defined by \eqref{eq:anal-feynman} is well--defined on a normal neighbourhood of the diagonal and weakly meromorphic in $\alpha$.

\begin{propo}
Consider a normal neighbourhood $\N_2$ of the diagonal $d_2\in\M^2$. The following statements hold for the analytically continued Feynman propagator $H^{(\alpha)}_F\in\D^\prime(\N_2)$ defined in \eqref{eq:anal-feynman}.
\begin{enumerate}
\item $\lim_{\alpha\to 0}H^{(\alpha)}_F= H_F$. 
\item $\WF(H^{(\alpha)}_F) \subset \WF(H_F)$.
\item The scaling degree of $H^{(\alpha)}_F$ tends to $-\infty$ when the real part of $\alpha$ tends to $\infty$.  
\end{enumerate}
\end{propo}
\begin{proof}
The proof of this proposition follows from the properties of $\sigma_F^{1+\alpha}$ obtained in Proposition \ref{pr:sigma-1}. In particular, $a)$ and $c)$ can be directly obtained from the weak analyticity, while $b)$ follows from the fact that the distribution $1/\sigma_F^\alpha$ is well defined on $\N_2\setminus d_2$ where it coincides either with $1/\sigma_+^\alpha$ or with $1/\sigma_-^\alpha$. 
In order to analyse the wave front sets of $1/\sigma_\pm^\alpha$, we pass to a normal coordinate system and obtain $1/\sigma_\pm^\alpha=2/(\xi^a\xi_a\pm i\epsilon \xi^0)^\alpha$. This distribution can be extended to a tempered  distribution for every $\alpha$ and thus its Fourier transform can be directly computed. One finds that for $1/\sigma_\pm^\alpha$, only the null future/past directed directions do not decay rapidly, consequently $H_F^{(\alpha)}$ restricted to $\N_2\setminus d_2$ has $\WF(H^{(\alpha)}_F)\subset \WF(H_F)$. Finally, we observe that the extension of $H^{(\alpha)}_F$ to $\N_2$ may possess further singularities supported on the diagonal with singular directions orthogonal to $d_2$. Hence, $\WF(H^{(\alpha)}_F)\subset \WF(H_F)$ still holds for $H_F^{(\alpha)}\in\D^\prime(\N_2)$.
\end{proof}

We are now able to discuss the analytical regularisation $\tau^{(\balpha)}_\Gamma$ \eqref{def:regularisedamplitudes} of the distributions $\tau_
\Gamma$ given in \eqref{eq:tgamma} which appear in the graph expansion  \eqref{eq:tau-gamma} of the time--ordered products $\mathcal{T}_n$ \eqref{eq:time-ordered-product}. As anticipated in Section \ref{sec:analytic_general}, owing to the form of $H^{(\alpha)}_F$ given in \eqref{eq:anal-feynman} the relevant distributions which need to be discussed are $t^{(\balpha)}_\Gamma$ introduced in \eqref{def:amplitudesigma} and analysed in the following proposition.

\begin{propo}\label{pr:prod-sigma}The operation
\begin{equation}\label{eq:prod-sigma}
\left\langle t_\Gamma^{(\balpha)}, f \right\rangle :=
\int_{\M^n}
\prod_{1\leq i < j \leq n } \frac{1}{\sigma_F(x_i,x_j)^{l_{ij}(1+ \alpha_{ij})}}    f dx_1 \dots dx_n
\end{equation}
defined for $f\in \D(\M^n\setminus D_n\cap \N)$ 
where $\N$ is a normal neighbourhood of the total diagonal (cf. Remark \ref{rem:geodesicneighbourhood}) has the following properties.
\begin{enumerate}
\item $t_\Gamma^{(\balpha)}$ is distribution on $\M^n\setminus D_n\cap \N$.
\item $\left\langle t_\Gamma^{(\balpha)}, f \right\rangle$ is a continuous function for $\balpha = \{\alpha_{ij}\}_{i<j} \in \mathbb{R}^{n(n-1)/2}$.
\item $\left\langle t_\Gamma^{(\balpha)}, f \right\rangle$ is analytic for every $\alpha_{ij}$ with $i<j$ and thus a multivariate analytic function.
\end{enumerate}
\end{propo}
\begin{proof}

$a)$ The domain $\M^n\setminus D_n\cap \N$ is a disjoint union of connected components. On every connected component $\mathcal{C}$ $\sigma_F(x_i,x_j)$ equals either $\sigma_+(x_i,x_j)$ or $\sigma_+(x_j,x_i)$ depending on the causal relation between $x_i$ and $x_j$ which is fixed in $\mathcal{C}$.
Hence, on $\mathcal{C}$, the wave front set of $\sigma_F(x_i,x_j)^{-1}$ is contained either in $\mathcal{V}_+$ or $\mathcal{V}_-$, where $\mathcal{V}_{+/-} = \{(x,x',k,k') \in T^*\M^2\setminus 0, (x,k)\sim (x',-k'),  k \triangleleft / \triangleright 0  \}$. Consequently, $\sigma_F(x_i,x_j)^{-1}$ satisfies the Hadamard condition up to a permutation of the arguments. The very same holds for the distributions $\sigma_F(x_i,x_j)^{l_{ij}(1+\alpha_{ij})}$ for every $l_{ij}$ and every $\alpha_{ij}$ which have been discussed in Proposition \ref{pr:sigma-1}.

Owing to the form of their wave front set, the pointwise products of these distributions present in $t_\Gamma^{(\balpha)}$ are well--defined because the H\"ormander--criterion for multiplication of distributions is satisfied. In fact, up to some fixed permutation of the arguments $(x_1,\dots x_n)$, $t_\Gamma^{(\balpha)}$ satisfies the micro local spectrum condition introduced in \cite{bfk:1996}. 
Hence $t_\Gamma^{(\balpha)}$ is a well--defined distribution on every connected component $\mathcal{C}$ of $\M^n\setminus D_n\cap \N$ and thus it is well--defined also on $\M^n\setminus D_n\cap \N$. 

$b)$ 
In order to check continuity for $\balpha= \{\alpha_{ij}\}_{i<j}\in \mathbb{R}^{n(n-1)/2}$ in a fixed point $\overline{\balpha}$ we may analyse the distribution on a fixed connected component $\mathcal{C}$ of the domain of $t_\Gamma^{(\balpha)}$ and factorize the distribution in two parts. 
In fact, due to the wave front set of $t_\Gamma^{(\balpha)}$ on $\mathcal{C}$ the factorisation $t_\Gamma^{(\balpha)} =t_\Gamma^{(\overline{\balpha})} \cdot \tau_\Gamma^{(\bbeta)}$ is unique where the integral kernel of $\tau_\Gamma^{(\bbeta)}$ is $\prod_{1\leq i < j \leq n } \frac{1}{\sigma_F(x_i,x_j)^{\beta_{ij}}}$.
For $\bbeta$ in a sufficiently small neighbourhood of $0$, $\tau_\Gamma^{(\bbeta)}$ is an integrable function which is differentiable for $\bbeta=0$ as can be obtained by dominated convergence. Finally, the continuity is preserved by pointwise multiplication with $t_\Gamma^{(\overline{\balpha})}$.

$c)$ 
For an arbitrary but fixed pair of indices $i,j$, $\alpha_{ij}$ appears in the product displayed in \eqref{eq:prod-sigma} as $1/\sigma_F(x_1,x_j)^{l_{ij}(1+\alpha_{ij})}$ and we have already analysed the analyticity property of such a distribution in Proposition \ref{pr:sigma-1}. 
We shall thus interpret $t_\Gamma^{(\balpha)}$ as a composition of distributions, namely as $1/\sigma_F^{\alpha_{ij}}\circ z $
 where $z$ is an operator which maps $\D(\M^{n}\setminus D_{n}\cap \N)$ to $ \D^\prime(\M^{2}\setminus D_2\cap \N_2)$ for a suitable $\N_2\supset D_2=d_2$. The $\epsilon$--regularised integral kernel of $z$ corresponds to the product present in \eqref{eq:prod-sigma} with the factor  $1/\sigma_F^{\alpha_{ij}}$ removed. Because of the singular structure of $z$, for every $f\in \D(\M^{n}\setminus D_{n}\cap \N)$, $\langle z,f\rangle$ is in fact a compactly supported smooth function supported on $\M^{2}\setminus D_2\cap \N_2$. Hence, the analysis of its composition with $1/\sigma_F^{\alpha_{ij}}$ is straightforward. These considerations imply separate analyticity of $t_\Gamma^{(\balpha)}$ in each $\alpha_{ij}$ whereas joint analyticity follows from the continuity proved in $b)$.
\end{proof}

\subsection{Generalised Euler operators and principal parts of homogeneous expansions}

\label{sec_R}

The next step in the strategy outlined at the end of Section \ref{sec:analytic_general} is to extend the distributions $t^{(\balpha)}_\Gamma$, 
 which are a priori defined only outside of the union of all partial diagonals $D_n$ in a normal neighborhood $\mathcal{N}$ of the total diagonal (cf. Remark \ref{rem:geodesicneighbourhood}) to $D_n\cap \mathcal{N}$ and to show that this extension is weakly meromorphic in $\alpha_I$ upon setting $\alpha_{ij}=\alpha_I$ for all $i,j\in I\subset \{1,\ldots,n\}$. As anticipated, we shall prove this by using particular homogeneity properties of $t^{(\balpha)}_\Gamma$ with respect to the scaling transformations \eqref{eq:n-dim-scaling}. Even if $t^{(\balpha)}_\Gamma$ is not homogeneous in the strong sense of Definition \ref{def:homogeneous}, it has weaker homogeneity properties which are still strong enough in order to obtain the wanted results. In this section we analyse analytically regularised distributions satisfying this weaker homogeneity condition, provide sufficient conditions for this weaker homogeneity to hold and show how the principal part of a distribution of this type can be efficiently computed.

To this avail, we consider a normal neighbourhood $\N_n$ of the total diagonal $d_n$ (cf. Remark \ref{rem:geodesicneighbourhood}) and define the {\bf generalised Euler operator} $E_p:\D(\N_n)\to\D(\N_n)$ by
\beq\label{def:geneuler}
E_p f (x_1,\dots, x_n) := (-1)^p \left. \lambda^{p+4(n-1)} \frac{d^p}{d\lambda^p}\left( \lambda^{-4(n-1)}  f_\lambda(x)\right)\right|_{\lambda = 1},
\eeq
where the scaling transformation \eqref{eq:n-dim-scaling} is used. We then consider a family of distributions $t^{(\alpha)}\in \D_\Gamma^\prime(\N_n\setminus d_n)\subset\D^\prime(\N_n\setminus d_n)$ (cf. Definition \ref{def:scalingdegree}) defined for $\alpha$ in some neighbourhood $\OO$ of $0\in \mathbb{C}$ and assume that $t^{(\alpha)}$ can be expanded as
\[
t^{(\alpha)}  = \sum_{k=0}^m t^{(\alpha)}_k + r^{(\alpha)}.
\]
where $t^{(\alpha)}_k$ are homogeneous with degree $a_k=-\delta_\alpha+k$ whose real part is smaller or equal to $-4(n-1)$ and a remainder $r^{(\alpha)}\in \D^\prime(\N_n\setminus d_n)$ which has scaling degree smaller than $4(n-1)$ and can thus be uniquely extended to $d_n$ for every $\alpha\in \OO$ by \cite[Theorem 6.9]{Brunetti-Fredenhagen:2000}, see also the discussion after Definition \ref{def:scalingdegree}. Owing to its homogeneity, every $t^{(\alpha)}_k$ can be rewritten by means of the generalised Euler operator $E_p$ as
\begin{equation}\label{eq:expose-poles0}
\left\langle t^{(\alpha)}_k, f \right\rangle   =   \frac{1}{\prod_{j=0}^{p-1} (a_k+j+4(n-1))}   \left \langle t^{(\alpha)}_k, E_p f \right\rangle.
\end{equation}

Note that, $E_p f(x_1, \dots x_n)$ is smooth and vanishes for $y=(x_1, \dots x_n) \to x=(x_1, \dots, x_1)$ as $C|y-x|^p$, i.e. it is in the class $O(|y-x|^{p})$. For this reason, if  $p$ is chosen sufficiently large as $p > -a_k-4(n-1)$, $t^{(\alpha)}_k\circ E_p$ possesses a unique extension to $d_n$. We recall that, in order to renormalise $t^{(\alpha)}$ for $\alpha=0$ in the MS--scheme, we have to subtract its principal part before computing the limit of vanishing $\alpha$
\[
\langle (t_k)_\ms, f \rangle:=\lim_{\alpha\to 0 } \left(\left\langle t^{(\alpha)}_k, f \right\rangle - \pp \left\langle t^{(\alpha)}_k, f \right\rangle  \right).
\]
However, if we use the representation of $t^{(\alpha)}_k$ provided by the right hand side of equation \eqref{eq:expose-poles0}, its poles are manifestly exposed and can be easily subtracted.
We recall that, since the original distribution $t^{(\alpha)}_k$ is well defined on $\N_n\setminus d_n$ even for $\alpha=0$, the principal part we are subtracting can only be supported on $d_n$. We summarise this discussion in the following proposition.

\begin{propo}\label{pr:regularisation} Consider a normal neighbourhood $\N_n$ of the total diagonal $d_n$ and a distribution $t\in\D^\prime(\N_n\setminus d_n)$. Assume that $t^{(\alpha)}\in\D^\prime(\N_n\setminus d_n)$ is an analytic regularisation of $t$, i.e. $t^{(\alpha)}$ is weakly analytic for $\alpha$ in a neighbourhood $\OO$ of the origin of $\bbC$ and $\lim_{\alpha\to 0}t^{(\alpha)}=t$. Moreover, assume that $t^{(\alpha)}$ can be decomposed as 
\[
t^{(\alpha)} = \sum_{k=0}^m t^{(\alpha)}_k    + r^{(\alpha)}
\]
where $t^{(\alpha)}_k$ are weakly analytic distributions which scale homogeneously under transformations of the form \eqref{eq:n-dim-scaling} with degree $a_k = -\delta_\alpha + k$ and $r^{(\alpha)}_k$ is a weakly analytic distribution whose scaling degree towards $d_n$ is strictly smaller than $4(n-1)$. Then, if $\WF(t)$ and $\WF(t^{(\alpha)})$ do not intersect the conormal bundle of the orbits of the scaling transformation \eqref{eq:n-dim-scaling}, the following statements hold.
\begin{enumerate}
\item $t^{(\alpha)}$ can be extended to $\dot{t}^{(\alpha)}\in \D^\prime(\N_n)$ for every $\alpha\in \OO\setminus\{0\}$.
\item $\dot{t}^{(\alpha)}$ is weakly meromorphic for $\alpha\in\OO$ with possible poles for $\alpha=0$ and it is the unique weakly meromorphic extension of $t^{(\alpha)}$. 
\item The pole of $\dot{t}^{(\alpha)}$in $0$ is supported on $d_n$.
\item The limit $\alpha\to 0$ can be considered after subtracting the pole part, namely
\[
\langle t_\ms, f \rangle:= \lim_{\alpha\to 0 } \left(\left\langle \dot{t}^{(\alpha)}, f \right\rangle - \pp \left\langle \dot{t}^{(\alpha)}, f \right\rangle  \right)
\]
is well--defined for all $f\in\D(\N_n)$ and $t_\ms$ is an extension of $t$ which preserves the scaling degree.
\end{enumerate}
\end{propo}

\begin{proof}
The proof of $a)$ and $b)$ follows from the discussion after Definition \ref{def:homogeneous} as an application of \cite[Theorem 3.2.3]{Hormander} to every $t^{(\alpha)}_k$ restricted to every orbit of the action of the scaling transformation \eqref{eq:n-dim-scaling}. Furthermore, since the scaling degree of $r^{(\alpha)}$ towards $d_n$ is strictly smaller than $4(n-1)$, $r^{(\alpha)}$ possesses an unique extension towards $d_n$, cf. \cite[Theorem 6.9]{Brunetti-Fredenhagen:2000}.

In order to prove $c)$ we note that the original distribution $t^{(\alpha)}$ defined on $\N_n\setminus d_n$ is weakly analytic and that an explicit construction of the weakly meromorphic extension $\dot{t}^{(\alpha)}$ to $\N_n$ is provided by 
\eqref{eq:expose-poles0}, choosing for every component $t^{(\alpha)}_k$ a sufficiently large $p$ and using again \cite[Theorem 6.9]{Brunetti-Fredenhagen:2000}. Hence, the poles of $\dot{t}^{(\alpha)}$ can only be supported on $d_n$. For this reason, after subtracting the principal part of the distribution the limit $\alpha \to 0$ can be safely taken. The such obtained distribution prior to considering the limit $\alpha \to 0$ coincides with $t^{(\alpha)}$ on $\N_n\setminus d_n$ and the same holds in the limit $\alpha\to 0$. Consequently $t_\ms$ is an extension of $t$. Finally, $\sd(t_\ms)=\sd(t)$, because our assumptions and the above analysis imply that $\sd(\dot{t}^{(\alpha)})=\sd(t^{(\alpha)})=\Ree(\delta_\alpha)$, $\sd(\pp(\dot{t}^{(\alpha)}))\le \Ree(\delta_\alpha)$ and $\lim_{\alpha\to 0}\Ree(\delta_\alpha) = \sd(t)$.
\end{proof}

We now discuss how equation \eqref{eq:expose-poles0} can be used in order to regularise the most singular part of a distribution 
$t^{(\alpha)}$ which is known to be of the form $t^{(\alpha)} = \sum_{k=0}^m t^{(\alpha)}_k + r^{(\alpha)}$ but where the distributions $ t^{(\alpha)}_k$ are not explicitly known. To this end, observe that equation \eqref{eq:expose-poles0} implies
\[
\left\langle t^{(\alpha)}, E_p f \right\rangle = 
\sum_{k=0}^m \left(\prod_{j=0}^{p-1} (a_k+j+4(n-1)) \right)  \left\langle t^{(\alpha)}_k, f \right\rangle + \left\langle r^{(\alpha)},E_p f\right\rangle.
\]
Moreover, we may assume without loss of generality as in Proposition \ref{pr:regularisation} that the homogeneity degrees $a_k$ of $t^{(\alpha)}_k$ are of the form $a_k=-\delta_\alpha + k$ where $\Ree(\delta_\alpha)$ is the scaling degree of $t^{(\alpha)}$. Consequently, $t^{(\alpha)}_0$ is the contribution with the highest scaling degree which may be extracted by introducing the coefficients
\[
c_k:=\prod_{j=0}^{p-1} \left(a_k+j+4(n-1)\right)
\]
and considering 
\begin{equation}\label{eq:decrease-scaling-degree}
\left\langle t^{(\alpha)}, E_p f \right\rangle - c_0   \left\langle t^{(\alpha)}, f \right\rangle   =
\sum_{k=1}^m (c_k-c_0)   \left\langle t^{(\alpha)}_k, f \right\rangle +  \left\langle r^{(\alpha)}, E_p f \right\rangle - c_0 \left\langle r^{(\alpha)}, f \right\rangle,
\end{equation}
where the distribution on the right hand side has a scaling degree smaller than $\Ree(\delta_\alpha) = - \Ree (a_0)$. Hence, although in general the distribution $t^{(\alpha)}$ does not scale homogeneously, equation \eqref{eq:expose-poles0} still holds up to distributions with a lower scaling degree. Knowing the decreasing degree of homogeneity of the components in the expansion of $t^{(\alpha)}$, we may use a recursive procedure in order to expose the pole part of this distribution. In fact, the previous discussion straightforwardly implies the validity of the following proposition.

\begin{propo}\label{pr:expose-poles}
We consider a distribution $t^{(\alpha)}$ with the properties assumed in Proposition \ref{pr:regularisation} and set
\[
u_0 := t^{(\alpha)},\qquad    u_{k+1} :=  c_k u_k -  u_k \circ E_{p_k}\,  , \qquad 0\le l< m
\]
where $p_k$ are the smallest natural numbers chosen in such a way that $p_k+\Ree (a_k)+4(n-1)>0$ and $c_k :=\prod_{j=0}^{p_k-1} \left(a_k+j+4(n-1)\right)$.
Then, in order to expose the poles of $t^{(\alpha)}$, we may invert the recursive definition of $u_k$ obtaining
\begin{equation}\label{eq:expose-poles}
t^{(\alpha)} = \frac{1}{c_0} \left( u_0\circ E_{p_0} +  \frac{1}{c_1} \left( u_1 \circ E_{p_1} +\dots +\frac{1}{c_n}\left(  u_n \circ E_{p_n}+ u_{n+1}   \right)   \right)\right).
\end{equation}
\end{propo}

In order to be able use the previous results for our purposes, we provide in the next proposition a criterion which is sufficient to ensure that a generic distribution can be decomposed into the sum of a homogeneous distribution and a remainder with lower scaling degree. We shall use this criterion in order to prove that the distributions $t_\Gamma^{(\balpha)}$ defined in Proposition \ref{pr:prod-sigma}
have the desired property.

\begin{propo}\label{pr:set} Let $\N_n$ be a normal neighbourhood of the total diagonal $d_n$ and suppose that $t\in \D^\prime(\N_n)$ has scaling degree $s_1$ towards $d_n$ under transformations of the form \eqref{eq:n-dim-scaling} and that there exists an $\alpha$ with $-\Ree(\alpha)=s_1$ such that $t\circ(E_1+\alpha+4(n-1))$ has scaling degree $s_2 < s_1$. Then $t$ can be decomposed into the sum of a homogeneous distribution with degree $\alpha$ and a remainder with scaling degree smaller than or equal to $s_2$.
\end{propo}
\begin{proof}
We start by observing that, for every test function $f\in\D(\N_n)$, 
\[
F(\lambda, f) := \left\langle t, f_{1/\lambda} \right\rangle
\] 
is a continuous linear functional of $f$ which is smooth in $\lambda$ for $\lambda>0$. Moreover, since the scaling degree of $t$ is $s_1$, $\lambda^{a}F(\lambda, f)$ vanishes in the limit $\lambda \to 0$ for every $a > s_1$ and for every $f\in\D(\N_n)$.
Let us now consider 
\[
G(\lambda, f)  :=  \left\langle(-E_1+\alpha+4(n-1))t, f_{1/\lambda}\right\rangle. 
\]
$G(\lambda, \cdot)$ is again a family of distributions on $\N_n$ which depends smoothly on $\lambda$ for positive $\lambda$. Furthermore, $\lambda^{a}G(\lambda, f)$ vanishes in the limit $\lambda\to 0$ for every $a > s_2$ and every $f\in\D(N_n)$. Hence,  
$\lambda^{\alpha-1} G(\lambda, \cdot)$ tends to $0$ in $\D^\prime(\N_0)$ for $\lambda\to 0$ and, additionally, the Banach--Steinhaus theorem implies that 
\begin{equation}\label{eq-bound-dist}
\left| \lambda^{a} G(\lambda,f) \right|  \leq C \sum_{\alpha\leq k} |\partial^{\alpha} f|,
\end{equation}
for every $a>s_2$, uniformly for $f$ supported in a compact set $K\subset \N_n$ and for suitable $C$ and $k$ which do not depend on $\lambda$.

After these preparatory considerations, we observe that $G$ and $F$ are related by the generalised Euler operator in the following way
\[
G(\lambda, f) =  \lambda^{-\alpha+1} \frac{d}{d\lambda} \lambda^{\alpha} F(\lambda, f).
\]
We can invert this relation to obtain
\[
F(\lambda, f) =  \frac{C(f)}{\lambda^{\alpha}} + \frac{1}{\lambda^\alpha} \int_0^\lambda  \tilde{\lambda}^{\alpha-1} G(\tilde{\lambda}, f) d\tilde{\lambda},
\]
where, $C(f)$ is a suitable constant which depends on $f$. We want to prove that $C(\cdot)$ is in fact a distribution. To this end, we note that, owing to the bound \eqref{eq-bound-dist}, the integral in $\tilde\lambda$ can be performed and the result of this integration is a distribution for every $\lambda > 0$ because $\Ree(\alpha) > s_2$.  
This implies that 
\[
C(f) = F(1, f)  - \int_0^1  \tilde{\lambda}^{\alpha-1} G(\tilde{\lambda}, f) d\tilde{\lambda}
\]
is a distribution because it is a linear combination of distributions. By construction $F(1,f_{1/\lambda})=F(\lambda,f)$ and $C\circ (E_1+\alpha) =0$, hence $C$ is a homogeneous distribution of degree $\alpha$. By means of a direct computation we also find that the scaling degree of the remainder $F(1,f)-C(f)$ is smaller than or equal to the scaling degree of $G$ which is $s_2$.
\end{proof}

\subsubsection{The differential form of generalised Euler operators and homogeneous expansions of Feynman amplitudes}
\label{sec:differentialEuler}

In order to make the previous discussion operative, we have to analyse the action of the generalised Euler operators $E_p$ appearing in \eqref{eq:expose-poles0} on test functions. In fact, we shall see that $E_p$ corresponds to a particular geometric partial differential operator. To this end, we observe that $E_p = (E_1-(p-1)) E_{p-1}$. Hence, knowing the differential form of the generalised Euler operator $E_1$, it is possible to construct recursively every $E_p$.

Regarding the differential form of $E_1$, we note that it can be written in terms of the geodesic distance and the van Vleck--Morette determinant\footnote{Recall that the square--root of the van Vleck--Morette determinant coincides with the Hadamard coefficient $u$ appearing in \eqref{eq:hadamard}.} $u^2$  as 
\[
E_1 f(x_1, \dots, x_n) = \sum_{j=2}^n \left(\sigma^a(x_j) \nabla^{x_j}_a  - \left(2  \sigma^a(x_j) \nabla^{x_j}_a  \log (u(x_j,x_1))\right)\right) f(x_1, \dots, x_n)\,,
\]
where $\nabla^{x_j}_a$ indicates the $a$--th component of the covariant derivative computed in $x_j$ and 
$\sigma^a(x_j) := {\nabla^{x_j}}^a\sigma(x_1,x_j)$.
Considering the adjoint $E^\dagger_p$ of $E_p$, we have $t\circ E_p = E^\dagger_p t$ where, using the relation $\Box \sigma + 2\sigma^a\nabla_a \log (u) = 4$, we find for $p=1$
\begin{eqnarray}
E_1^\dagger  t(x_1,\dots,x_n)     
&=&  \sum_{j=2}^n \left(- \nabla^{x_j}_a \sigma^a(x_j)    - 2 \sigma^a(x_j) \left(\nabla^{x_j}_a \log (u(x_j,x_1)) \right)\right) t(x_1,\dots,x_n) \notag \\
&=& -\left( 4(n-1) +  \sum_{j=2}^n \sigma^a(x_j) \nabla^{x_j}_a   \right)t(x_1,\dots,x_n).
\label{eq:euler-operator}
\end{eqnarray}
We finally observe that the recursive identity for $E_p$ implies that also $E^\dagger_p$ can be constructed recursively starting from $E^\dagger_1$ as 
\beq\label{eq:euler-operator2}
E_p^\dagger =  E_{p-1}^\dagger (E_1^\dagger-(p-1)).
\eeq

We proceed by showing that upon applying $E^\dagger_1$ introduced in  \eqref{eq:euler-operator} to a distribution $t_\Gamma^{(\balpha)}$ of the form
\[
t_\Gamma^{(\balpha)}=\prod_{1\leq i < j \leq n } \frac{1}{\sigma_F(x_i,x_j)^{l_{ij}(1+ \alpha_{ij})}}
\]
which has scaling degree $\sd(t_\Gamma^{(\balpha)}) = \sum_{i<j} 2 l_{ij}(1+ \Ree(\alpha_{ij}))$ towards the thin diagonal $d_n$, the result is a term proportional to $t_\Gamma^{(\balpha)}$ plus a remainder which has lower scaling degree as foreseen in \eqref{eq:decrease-scaling-degree}. Hence, Proposition \ref{pr:set} implies that $t_\Gamma^{(\balpha)}$ can be written as a homogeneous distribution plus a remainder with lower scaling degree. If the scaling degree of the remainder is not sufficiently low, we reiterate the procedure in order to obtain a full almost homogeneous expansion of the desired form.

In order to analyse this issue we shall only consider the relevant differential operator on $\M^n$ appearing in $E_1^\dagger$, namely,
\begin{equation}\label{eq:rho}
\rho:= -  \sum_{j=2}^n \sigma^a(x_j) \nabla^{x_j}_a.    
\end{equation}
We start by analysing the action of $\rho$ on $\sigma(x_2,x_3)$ for $x_2,x_3$ in a normal neighbourhood of the point $x_1$.

\begin{lem}\label{le:rho-over-squared}Let $\N_{x_1}$ be a normal neighbourhood of the point $x_1$ and let $x_2,x_3\in\N_{x_1}$.
Then,  
\[
\rho \sigma(x_2,x_3) = 2\sigma(x_2,x_3) + G(x_1,x_2,x_3)    
\]
where $G$ is a smooth function which vanishes in the limit $x_2,x_3\to x_1$ as a monomial of order $4$ in the normal coordinates of $x_2$ and $x_3$ centred in $x_1$. 
\end{lem}

\begin{proof}Using the notation in the proof of Proposition \ref{pr:sigma-1} we write the action of $\rho$ on $\sigma_{23}:=\sigma(x_2,x_3)$ as
\[
\rho \sigma_{23}  =    \xi_a(x_2) \sigma^a_{23} + \xi_{b^\prime}(x_3)\sigma^{b^\prime}_{23}\,.
\]
Recall that $\sigma_{23}^a$ is the covector in $T^*_{x_2}M$ cotangent to the unique geodesic joining $x_2$ and $x_3$, that $-\sigma^{b'}_{23}$ is equal to the parallel transport of $\sigma^a_{23}$ from $x_2$ to $x_3$ along the geodesic $\gamma$ joining the two points, and that $\xi^c(x_i):=\sigma^c(x_1,x_i)$.

Let us parametrise the image of $\gamma$ with an affine parameter $\lambda$ such that $x(0) = x_2$ and $x(1) = x_3$. In order to simplify the notation, we indicate by $t(\lambda)$ the tangent vector of the geodesic in $x(\lambda)$. As argued before, we have
\[
t^a(0)=\sigma^a_{23},\qquad \text{and}\qquad t^{b^\prime}(1)=-\sigma^{b^\prime}_{23}.
\]   
Consequently,
\begin{eqnarray*}
\rho \sigma_{23}
&=& \xi^a t_a (0) - \xi^b t_b(1) 
= - \int_{0}^{1} \frac{d}{d\lambda} (\xi^a t_a)(\lambda) d\lambda \\
&=& - \int_{0}^{1} t^a \nabla_t  \xi_a   d\lambda  
= \int_{0}^{1} t^a t^b \sigma_{ab}(x(\lambda),x_1) d\lambda\,,
\end{eqnarray*}
where $\sigma_{ab} := \nabla_a\nabla_b \sigma$. If we now consider the covariant Taylor expansion of $\sigma_{ab}(x(\lambda),x_1)$ around $x(\lambda)$ (see e.g. \cite{Poisson:2011nh}), we find that $E_{ab}(x,x_1):=\sigma_{ab}(x,x_1) - g_{ab}(x) $ is a smooth function that vanishes for $x\to x_1$ as $O(\sigma(x,x_1))$, hence
\[
\rho\sigma_{23}  = \int_{0}^{1} t^a t^b g_{ab}(x(\lambda))    d\lambda +  \int_{0}^{1} t^a t^b E_{ab}(x(\lambda),x_1)    d\lambda  = 2\sigma(x_2,x_3) + G(x_1,x_2,x_3)\,,
\]
where the remainder is smooth because of the smoothness of the metric $g$ and can be further expanded as 
\begin{eqnarray}\label{eq:remainder-sigma}
G(x_1,x_2,x_3) &=& \int_{0}^{1} t^a(\lambda) t^b(\lambda) \left(\sigma_{ab}(x(\lambda),x_1) - g_{ab}(x(\lambda))\right) d\lambda\\ &=&  
\int_{0}^{1} t^a(\lambda) t^b(\lambda) t^c(\lambda) t^d(\lambda) R_{acbd}(x(\lambda))d\lambda + \dots = O(|\xi(x_2)|^4+|\xi(x_3)|^4)\,,\notag
\end{eqnarray}
where the absolute value of the normal coordinates $|\xi(x_i)|$ of $x_i$, $i=2,3$ is intended in the Euclidean sense.
\end{proof}

We are now in position to analyse the action of $\rho$ on the distribution $t_\Gamma^{(\balpha)}$ introduced in \eqref{eq:prod-sigma}.

\begin{propo}\label{pr:almost-homo}
The distribution $t_\Gamma^{(\balpha)}$ introduced in \eqref{eq:prod-sigma} can be written as a sum of homogeneous distributions with respect to scaling towards the total diagonal $d_n$ plus a remainder. The degrees of homogeneity of these homogeneous distributions are contained in the following set
\[
\left\{k-\sum_{1\leq i<j\leq n} 2 l_{ij}(1+ \alpha_{ij}), k\in \mathbb{N}\cup \{0\} \right\}.
\] 
\end{propo}

\begin{proof}
We perform this analysis with $\epsilon$ in $\sigma_F$ taken to be strictly positive. We start by applying $\rho$ given in \eqref{eq:rho} to $t_\Gamma^{(\balpha)}$. Thanks to the results stated in Lemma \ref{le:rho-over-squared} we have
\[
\rho t_\Gamma^{(\balpha)}  =    C   t_\Gamma^{(\balpha)}    + r^{(\balpha)}_\Gamma,
\]
where the constant $C$ is
\[
C= -\sum_{1\leq i<j\leq n} 2 l_{ij}(1+ \alpha_{ij}).
\]
Furthermore, Lemma \ref{le:rho-over-squared} and in particular \eqref{eq:remainder-sigma} implies that the remainder $r^{(\balpha)}_\Gamma$ has a scaling degree towards $d_n$ which is lower than the one of $t_\Gamma^{(\balpha)}$ by at least two, 
\begin{equation}\label{eq:sd-tgamma}
\sd(r^{(\balpha)}_\Gamma) \leq \sd(t_\Gamma^{(\balpha)})-2 = \sum_{1\leq i<j\leq n} 2 l_{ij}(1+ \Ree(\alpha_{ij})) -2\,.
\end{equation}
Proposition \ref{pr:set} then implies that the distribution $t_\Gamma^{(\balpha)}$ can be written as a homogeneous distribution of degree $C$  
plus a remainder with lower scaling degree.

In order to finalise the proof we need to control the recursive application of $\rho$, therefore we discuss the application of $\rho$ on $\rho^nt_\Gamma^{(\balpha)}$ for an arbitrary $n$. 
Let us start with $n=1$. In this case, we observe that the relevant contribution is the one given by the remainder $\rho r^{(\balpha)}_\Gamma$, which reads 
\[
r^{(\balpha)}_\Gamma = \sum_{1\leq i < j \leq n} l_{ij}(1+ \alpha_{ij}) \frac{G(x_1,x_i,x_j)}{\sigma_F(x_i,x_j)}t_\Gamma^{(\balpha)}.
\]
Note that for every $i<j$, $\sigma_F(x_i,x_j)t_\Gamma^{(\balpha)}$ has the same structure like $t_\Gamma^{(\balpha)}$, but the scaling degree $\sd(t_\Gamma^{(\balpha)}) +2$, whereas $G(x_1,x_i,x_j)$ defined in \eqref{eq:remainder-sigma} is a smooth function whose Taylor expansion for $x_i, x_j$ around $x_1$ starts with components of order $4$. Hence, if we apply $\rho$ to $r^{(\balpha)}_\Gamma$ we obtain
a constant multiple of $r^{(\balpha)}_\Gamma$ plus a remainder which has scaling degree lower or equal to $\sd(r^{(\balpha)}_\Gamma)-1$, where the difference with respect to \eqref{eq:sd-tgamma} stems from the fact that $G$ can be expanded as a polynomial in $\sigma_a(x_i)$ whose lowest components are monomials of degree $4$ multiplied by curvature tensors. These monomials are homogeneous and thus contribute to the degree of homogeneity of $\rho r^{(\balpha)}_\Gamma$, while the contributions in $G$ with degree higher or equal to five influence the scaling degree of the remainder.
Repeating this analysis for a generic $n$, we find that similar results hold when $\rho$ is applied recursively to the remainder.

Consequently, an iterated application of Proposition \ref{pr:set} implies that the distribution $t_\Gamma^{(\balpha)}$ can be written as a finite sum of homogeneous distributions plus a remainder. Furthermore, since the scaling degree of these distributions is always finite, the degree of homogeneity of these components is finite as well. 
\end{proof}

As outlined at the end of Section \ref{sec:analytic_general}, we can use Proposition \ref{pr:almost-homo} in conjunction with the propositions \ref{pr:regularisation} and \ref{pr:expose-poles} in order to extend the distributions $t_\Gamma^{(\balpha)}$ in a unique and weakly meromorphic fashion to the union of all partial diagonal $D_n$ in a normal neighborhood of the total diagonal (cf. Remark \ref{rem:geodesicneighbourhood}) and in order to compute the relevant pole part of this extension as used in the forest formula, cf. \eqref{eq:forset-formula}, 
 \eqref{def:regularisedamplitudes} and \eqref{eq:anal-feynman}. To this avail, we stress that Proposition \ref{pr:almost-homo} holds in particular for any subgraph $\Gamma_I$, $I\subset\{1,\dots,n\}$ of $\Gamma$ and the corresponding distribution $t_{\Gamma_I}^{(\balpha)}$ which is obtained by omitting all factors in $t_{\Gamma}^{(\balpha)}$ which correspond to edges not contained in $\Gamma_I$. Finally, the recursive structure of the forest formula \eqref{eq:forset-formula} implies that we are not dealing only with expressions of the form $t_{\Gamma_I}^{(\balpha)}|_{\alpha_{ij}=\alpha_I \forall i,j\in I}$, but also with expressions which are of this form up to a subtraction of their principal part. However, our above analysis and in particular the discussion in the proof of Proposition \ref{pr:regularisation} implies that the propositions \ref{pr:almost-homo} and \ref{pr:expose-poles} also hold in this case.

\begin{rem}\label{rem:differential}
Proposition \ref{pr:expose-poles} and the above analysis imply that our renormalisation scheme is in fact a particular form of differential renormalisation. Notwithstanding, the advantage of formulating this scheme in terms of analytic regularisation and minimal subtraction is the ability to define the renormalisation scheme in a closed form at all orders by means of the forest formula \eqref{eq:forset-formula}.
\end{rem}

\subsection{Properties of the minimal subtraction scheme}

We conclude the general analysis of the renormalisation scheme introduced in this work by demonstrating that this scheme satisfies -- up to one property we shall mention at the end of this section -- all axioms of \cite{Hollands:2001nf,Hollands:2001b,  Hollands:2004yh} which, as argued in these works, any physically meaningful scheme to renormalise time--ordered products should satisfy. We refer to these works for a detailed formulation and discussion of these axioms. In addition to showing these properties of the scheme, we also argue that it preserves invariance under any spacetime isometries present.

\begin{propo}\label{pr:propertiesscheme}
The time--ordered product $\mathcal{T}_n$ defined by means of \eqref{eq:forset-formula}, where the quantities appearing in this formula are defined by means of \eqref{eq:projection}, \eqref{eq:tau-gamma_reg}, \eqref{def:regularisedamplitudes} and \eqref{eq:anal-feynman}, and were we recall Remark \ref{rem:geodesicneighbourhood}, have the following properties.
\begin{enumerate}
\item $\mathcal{T}_n$ is symmetric and satisfies the causal factorisation condition.
\item $\mathcal{T}_n$ is unitary.
\item $\mathcal{T}_n$ is local and covariant.
\item $\mathcal{T}_n$ satisfies the microlocal spectrum condition.
\item $\mathcal{T}_n$ is $\phi$--independent.
\item $\mathcal{T}_n$ satisfies the Leibniz rule.
\item $\mathcal{T}_n$ satisfies the Principle of Perturbative Agreement for perturbations of the generalised mass term $\mu$ in the free Klein--Gordon equation $P\phi:=(-\Box + \mu)\phi=0$.
\item If the spacetime $(\M,g)$ has non--trivial isometries and if the Feynman propagator $H_F$ is chosen such as to be invariant under these isometries, then $\mathcal{T}_n$ is invariant under these isometries as well.
\end{enumerate}
\end{propo}
\begin{proof}
$a)$ holds because we constructed the renormalised time--ordered product by means of the forest formula \eqref{eq:forset-formula} and because, as implied by Proposition \ref{pr:regularisation}, all counterterms subtracted in the forest formula are local.

$b)$ Unitarity holds because the operation of extracting the relevant principal part of a regularised amplitude $\tau^{(\balpha)}_\Gamma$ commutes with complex conjugation (even if $\balpha$ is not real).

$c)$ The regularised amplitudes $\tau^{(\balpha)}_\Gamma$ satisfy locality and covariance. Upon setting $\alpha_{ij}=\alpha_I$ for $i,j\in I\subset \{1,\dots,n\}$, $\tau^{(\balpha)}_\Gamma$ is weakly meromorphic in $\alpha_I$. Thus locality and covariance holds for each term in the corresponding Laurent series and consequently also after subtracting the principal part of this series.

$d)$ As argued in the proof of Proposition \ref{pr:prod-sigma}, the distributions $t^{(\balpha)}_\Gamma$ defined in \eqref{eq:prod-sigma} satisfy the microlocal spectrum condition, i.e. they have the correct wave front set. Consequently, the regularised amplitudes $\tau^{(\balpha)}_\Gamma$ have the correct wave front set as well. As $\tau^{(\balpha)}_\Gamma$ is weakly meromorphic in the sense recalled in the proof of $c)$, each term in the corresponding Laurent series has a wave front set bounded by the wave front set of $\tau^{(\balpha)}_\Gamma$. Consequently the microlocal spectrum condition holds after subtracting the principal part and considering the limit of vanishing regularisation parameters.

$e)$ This property follows directly from the construction. In particular the subtraction of counterterms is defined in terms of numerical distributions and independent of the field $\phi$.

$f)$ In analogy to $b)$, the Leibniz rule holds because the operation of extracting the relevant principal part of a regularised amplitude $\tau^{(\balpha)}_\Gamma$ commutes with all partial differential operators.

$g)$ The Principle of Perturbative Agreement for perturbations of the generalised mass term $\mu$ demands essentially that upon setting $\mu=\mu_0 + \mu_1$, the renormalisation of $\mathcal{T}_n$ commutes with the operation of perturbatively expanding quantities in $\mu_1$ around $\mu_0$. A Feynman propagator $H_F$ depends on $\mu$ only via the Hadamard coefficients $v$ and $w$ in \eqref{eq:hadamard}. However, in the definition of the analytically regularised $H^{\alpha}_F$ in \eqref{eq:anal-feynman} and the corresponding regularised amplitudes $\tau^{(\balpha)}_\Gamma$ defined in \eqref{def:regularisedamplitudes}, these coefficients are not altered but only the $\sigma$--dependent terms multiplying these coefficients are modified. Consequently, the analytic regularisation and minimal subtraction scheme we consider commutes with a perturbative expansion in $\mu_1$ around $\mu_0$.

$h)$ As recalled in $g)$ all operations in our analytic regularisation and minimal subtraction scheme act directly on quantities defined entirely in terms of the geometric quantity $\sigma$. As $\sigma$ is invariant under any spacetime isometries present, the renormalisation scheme preserves this invariance.
\end{proof}

\begin{rem}\label{rem:missingPPA}
Note that the Principle of Perturbative Agreement (PPA) as introduced in  \cite{Hollands:2004yh} also poses conditions on $\mathcal{T}_1$, i.e. the renormalisation of local and covariant Wick polynomials, which we omitted in our analysis, cf. Footnote \ref{foot:T1} on page \pageref{foot:T1}. However, given $\mathcal{T}_n$ for $n>1$, $\mathcal{T}_1$ can be adjusted in order to satisfy the PPA for changes of $\mu$ by using e.g. \cite[Theorem 3.3]{DHP}. Moreover, the PPA as introduced in \cite{Hollands:2004yh} further demands that, setting $g=g_0+g_1$, the renormalisation also commutes with perturbatively expanding quantities in $g_1$ around an arbitrary but fixed background metric $g_0$. Since $\sigma$ depends on $g$, it is not easy to check whether a perturbative expansion in $g_1$ commutes with our analytic regularisation and minimal subtraction scheme and thus it might well be that the renormalisation scheme discussed in the present work fails to satisfy this part of the PPA. However, if this is the case, the scheme can be modified according to the construction in \cite{Hollands:2004yh} in order to satisfy also this condition while preserving the other properties in Proposition \ref{pr:propertiesscheme}, including the invariance under any spacetime isometries present.
\end{rem}

\begin{rem}\label{rem:rengroup}
We have omitted the explicit dependence of renormalised quantities on the mass scale $M$ appearing in the analytically regularised Feynman propagator $H^{(\alpha)}_F$ \eqref{eq:anal-feynman}, but our analysis implies that the dependence of these quantities on $M$ is such that all renormalised quantities are polynomials of (derivatives of) $\log\left( M^2 \sigma_F(x_i,x_j)\right)$, see also the examples in the next section. Thus, the renormalisation group flow with respect to changes of $M$ may be easily computed.
\end{rem}

\subsection{Examples}
\label{sec_fishsunset}

In this section we illustrate the method developed in Section \ref{sec_R} to explicitly compute renormalised quantities in our scheme by considering first the example of the fish graph and the sunset graph, i.e. $\Delta^n_F$ for $n=2,3$. These pointwise powers of the Feynman propagator are the only ones occurring in renormalisable scalar field theories in four spacetime dimensions. Afterwards we will consider a triangular graph in Section \ref{sec:complicatedgraph} in order to illustrate the method in the case of more than two vertices. Recalling Remark \ref{rem:geodesicneighbourhood}, we shall work only on subsets of the spacetime where the geodesic distance is well--defined without loss of generality.

In the special case of $\Delta^n_F$, we are dealing with distributions which are already defined on $\M^2\setminus d_2$ and have to be extended to $\M^2$. In order to accomplish this task we shall use \eqref{eq:expose-poles0} in order to expose the poles before subtracting them. In this context, we note that  $E^\dagger_1$ given in \eqref{eq:euler-operator} applied to a distribution $t$ whose integral kernel $t(\sigma_F)$ depends on $x,y$  only via $\sigma_F(x,y)$, can be further simplified. 
In particular, introducing $t_1(\sigma_F)$ such that $\nabla^a t_1(\sigma_F) = \sigma^a t(\sigma_F)$, we have 
\begin{eqnarray}\label{eq:E-simplified}
E^\dagger_1 t(\sigma_F)&=& -\left( 4 +  \sigma^a\nabla_a   \right)t(\sigma)  = 
- \nabla_a \sigma^a  t - 2 \sigma^a (\nabla_a \log (u))   t(\sigma_F)  \\ &=&
- \Box t_1(\sigma_F)  - 2  \frac{\nabla_a u}{u}   \nabla^a t_1(\sigma_F)\,,\notag  
\end{eqnarray}
where $x$ is considered to be arbitrary but fixed and all the covariant derivatives are taken with respect to $y$.

\subsubsection{Computation of the renormalised fish and sunset graphs in our scheme}

We recall that the Feynman propagator $\Delta_F(x,y):=\langle\phi(x)\cdot_{T_\Delta}\phi(y)\rangle_\Omega$ in any Hadamard state $\Omega$ is locally of the form
\beq\label{eq_DeltaF}
\Delta_F(x,y)=\frac{1}{8\pi^2}\left(\frac{u(x,y)}{\sigma_F(x,y)}+v(x,y)\log(M^2 \sigma_F(x,y))\right)+w(x,y)\,,\qquad\sigma_F:=\sigma+i\epsilon\,.
\eeq
From \eqref{eq_DeltaF} we can infer that, in order to renormalise $\Delta^2_F$ and $\Delta^3_F$, i.e. in order to extend them from $\M^2\setminus d_2$ to $\M^2$, we need to renormalise the three distributions
\beq\label{eq_sigma_problematic}
\frac{1}{\sigma_F^2}\qquad \frac{\log \left(M^2 \sigma_F\right)}{\sigma_F^2}\qquad \frac{1}{\sigma_F^3}\,,
\eeq
because all other occurring powers of $\sigma_F$, i.e. $\sigma^{-m}_F\log^n (\sigma_F)$ for $m\in\{0,1\}$ and $n\in\{0,1,2,3\}$ have a scaling degree for $y\to x$ smaller than 4, and thus can be uniquely extended to the diagonal.
To this avail, we define
$$\sigma_{a_1\cdots a_n}:=\nabla_{a_n}\cdots\nabla_{a_1}\sigma\qquad [B](x):=B(x,x)\,,$$
where the covariant derivatives are taken with respect to $x$ and $B$ is a general bitensor, and recall the following basic identities satisfied by $\sigma$:
\begin{equation}\label{eq_basicidentities}
\sigma_a \sigma^a  = 2\sigma\,,\qquad \sigma_{ab}\sigma^b=\sigma_a\,,\qquad \Box\sigma = 4 - 2 \frac{\sigma^a \nabla_a u}{u}\,.
\end{equation}
For our purposes, it will prove useful to use the last identity in the form
\beq\label{eq_def_f}\Box \sigma_F = 4 + f \sigma_F\qquad\text{with}\qquad f:=- 2 \frac{\sigma^a \nabla_a u}{u\sigma_F}\,,\eeq
where $f$ is a distribution, which, considered as a distribution in $y$ for fixed $x$, has scaling degree zero for $y\to x$ as can be seen from the covariant Taylor expansion $u=[u]+\left([\nabla_a u]-\nabla_a[u]\right)\sigma^a + \R_u=1+\R_u$, where the remainder $\R_u$ vanishes towards the diagonal faster than $\sigma_a$  (see e.g. \cite[Section 5]{Poisson:2011nh}).

\begin{rem}\label{rem_fdists}
As $f$ has vanishing scaling degree for $y\to x$, the pointwise product $f(x,y) t(x,y)$ with any bidistribution $t$ of scaling degree for $y\to x$ lower than 4 may be uniquely extended to the diagonal. However, we will also encounter expressions which are naively of the form $f(x,y) \delta(x,y)$ and which are a priory ill--defined because $f$ is in general divergent for $x$ and $y$ light--like related, and thus not continuous on the diagonal. Notwithstanding, the distribution $f(x,y) \delta(x,y)$, which is well--defined and identically vanishing outside of the diagonal $x=y$, may be extended to the diagonal. In fact, our scheme, in which expressions of the form $f(x,y) \delta(x,y)$ appear as $\alpha\to 0$ limits of particular weakly analytic expressions, provides a unique and non--vanishing extension of $f(x,y) \delta(x,y)$ to the diagonal by the very analyticity of the aforementioned expressions. In particular our scheme implies the following unique and well--defined definitions of distributions on $\M^2$.
\beq\label{eq_fdists}
f\Box \frac{\log^n (M^2 \sigma_F) }{\sigma_F}:= \lim_{\alpha\to 0}f\Box \frac{\log^n (M^2 \sigma_F) }{\sigma^{1+\alpha}_F}\,,\qquad n\ge 0
\eeq
Hereby uniqueness and weak analyticity of $f\Box(\log^n (M^2 \sigma_F) /\sigma^{1+\alpha}_F)$ follow from arguments used throughout Section \ref{sec_R}.
\end{rem}

From Proposition \ref{pr:sigma-1}, we know that $1/\sigma^{n+\alpha}_F$ is weakly meromorphic in $\alpha$. In order to compute the Laurent series, we use the above--mentioned identities for $\sigma$ and obtain
$$
\frac{1}{\sigma^{n+1+\alpha}_F}=\frac{1}{2(n+\alpha)(n-1+\alpha)}\left(\Box+(n+\alpha)f\right)\frac{1}{\sigma^{n+\alpha}_F}\,
$$
in accordance with \eqref{eq:expose-poles0} and \eqref{eq:E-simplified}.

 Using this and recalling Remark \ref{rem_fdists}, we may compute the following Laurent series, where we recall that in $\Delta^{(\alpha)}_F$ \eqref{eq:anal-feynman} we use the same (arbitrary) constant $M$ present in the logarithmic term  of \eqref{eq_DeltaF} to correct for the change of dimension and a sufficiently regular function $k$ for later purposes,
\begin{align}\frac{1}{(Mk)^{2\alpha}}\frac{1}{\sigma^{2+\alpha}_F}=&\frac{1}{2}(\Box+f)\left(\frac{1}{\alpha\sigma_F}-\frac{\log \left(M^2 \sigma_F\right)}{\sigma_F}\right)-\frac{\log(k^2)}{2}(\Box+f)\frac{1}{\sigma_F}-\Box \frac{1}{2\sigma_F}+O(\alpha)\,,\notag\\
\label{eq_generalexpansion}\frac{d}{d\alpha}\frac{1}{(Mk)^{2\alpha}}\frac{1}{\sigma^{2+\alpha}_F}=&\frac12\left(\Box+f\right)\left(-\frac{1}{\alpha^2 \sigma_F}+\frac{\log^2 \left(M^2 \sigma_F\right)}{2\sigma_F}\right)+\Box\frac{\log \left(M^2 \sigma_F\right)+1}{2\sigma_F}\\&+\log^2(k^2)(\Box+f)\frac{1}{4\sigma_F}+\log(k^2)\left(\Box\frac{1}{2\sigma_F}+\left(\Box+f\right)\frac{\log \left(M^2\sigma_F\right)}{2\sigma_F}\right)+O(\alpha)\,,\notag\\
\frac{1}{(M h)^{2\alpha}}\frac{1}{\sigma^{3+\alpha}_F}=&\frac{1}{8}(\Box+2f)(\Box+f)\left(\frac{1}{\alpha\sigma_F}-\frac{\log \left(M^2 \sigma_F\right)}{\sigma_F}\right)-\frac{\log(k^2)}{8}(\Box+2f)(\Box+f)\frac{1}{\sigma_F}\notag\\&-\frac{1}{16} \left((5\Box+8f)(\Box+f)-2(\Box+2f)f\right) \frac{1}{\sigma_F}+O(\alpha)\,.\notag\end{align}
Note that by means of Lemma \ref{lem_productidentities} b) one may explicitly check that the pole terms in these Laurent series are local expressions as expected.

Using the Laurent series, the lowest renormalised powers of $\sigma_F$ may be defined and computed as\footnote{Note that we use here a definition of the analytic regularisation of the logarithm in terms of a direct derivative rather than a limit of differences like in \eqref{eq:anal-feynman}. While the two definitions differ up to a constant factor in the principal part, they coincide in the constant regular part and thus give the same $(
\sigma^{-2}_F \log (M^2 \sigma_F))_\ms$.}, where we recall once again Remark \ref{rem_fdists}.
\begin{align}\left(\frac{1}{\sigma_F^2}\right)_\ms:=&\lim_{\alpha\to 0}\left(\frac{1}{M^{2\alpha}}\frac{1}{\sigma^{2+\alpha}_F}-\pp\frac{1}{M^{2\alpha}}\frac{1}{\sigma^{2+\alpha}_F}\right)=-\frac{1}{2}(\Box+f)\frac{\log \left(M^2 \sigma_F\right)}{\sigma_F}-\Box \frac{1}{2\sigma_F}\,,\notag\\
\label{eq_sigma_ms}\left(\frac{\log \left(M^2\sigma_F\right)}{\sigma_F^2}\right)_\ms:=&-\lim_{\alpha\to 0}\left(\frac{d}{d\alpha}\frac{1}{M^{2\alpha}}\frac{1}{\sigma^{2+\alpha}_F}-\pp\frac{d}{d\alpha}\frac{1}{M^{2\alpha}}\frac{1}{\sigma^{2+\alpha}_F}\right)\\=&-\frac14\left(\Box+f\right)\frac{\log^2 \left(M^2 \sigma_F\right)}{\sigma_F}-\Box\frac{\log \left(M^2 \sigma_F\right)+1}{2\sigma_F}\,,\notag\\
\left(\frac{1}{\sigma_F^3}\right)_\ms:=&\lim_{\alpha\to 0}\left(\frac{1}{M^{2\alpha}}\frac{1}{\sigma^{3+\alpha}_F}-\pp\frac{1}{M^{2\alpha}}\frac{1}{\sigma^{3+\alpha}_F}\right)\notag\\=&-\frac{1}{8}(\Box+2f)(\Box+f)\frac{\log \left(M^2 \sigma_F\right)}{\sigma_F}-\frac{1}{16} \left((5\Box+8f)(\Box+f)-2(\Box+2f)f\right) \frac{1}{\sigma_F}\,.\notag\end{align}
Finally $\left(\Delta^2_F\right)_\ms$ and $\left(\Delta^3_F\right)_\ms$ are defined and computed by expanding the unrenormalised powers $\Delta^2_F$ and $\Delta^3_F$ and replacing the three problematic expressions \eqref{eq_sigma_problematic} by their renormalised versions \eqref{eq_sigma_ms}.

\subsubsection{Alternative computation of the renormalised fish and sunset graphs}

As a preparation towards the application of our renormalisation scheme to QFT in cosmological spacetimes, we shall now derive an alternative way to compute $\left(\Delta^2_F\right)_\ms$ and $\left(\Delta^3_F\right)_\ms$, which is better suited for practical computations. We start by stating and proving a few distributional identities.
\begin{lem}\label{lem_productidentities}The following distributional identities hold.
\begin{enumerate}
\item For any continuous $F_0$ and any twice continuously differentiable $F_2$, 
$$\sigma F_0 \delta = 0\,,\qquad \sigma_a F_0 \delta = 0\,,\qquad F_0\nabla_{\nabla\sigma}\delta =-[F_0 \Box \sigma]\delta\,,$$ 
$$F_2 \Box\delta=[\Box F_2]\delta+\Box [F_2]\delta - 2\nabla^a[\nabla_aF_2]\delta\,.$$
\item $$(\Box+f)\frac{1}{\sigma_F}=8\pi^2i\delta$$
$$(\Box+2f)(\Box+f)\frac{1}{\sigma_F}:=\lim_{\alpha\to 0}(\Box+2f)(\Box+f)\frac{1}{\sigma^{1+\alpha}_F}=8\pi^2i\left(\Box-\frac R3\right)\delta$$
\item For all $n_1$, $n_2$, $n_3\in\bbN_0$ and $n_4$, $n_5$, $n_6\in\{0,1\}$ with $n_2-n_3+n_4\ge-1$,  $$\log^{n_1}\!\!\left(\sigma_F\right)(\sigma^a_F)^{n_4} \sigma^{n_2}_F \left(\frac{1}{\sigma_F^{n_3}}\right)_\ms= \log^{n_1}\!\!\left(\sigma_F\right) (\sigma^a_F)^{n_4}\sigma^{n_2-n_3}_F\,,$$
$$\Box \log (\sigma_F) = \frac{\Box \sigma -2}{\sigma_F}\,,\qquad \nabla_a\frac{\log^{n_5}(\sigma_F)}{\sigma^{n_6}_F}=\frac{\left(n_5-n_6\log^{n_5} (\sigma_F)\right)\nabla_a \sigma}{\sigma^{n_6+1}_F}\,.$$ 
\item $$\sigma_F \left(\frac{1}{\sigma_F^3}\right)_\ms=\left(\frac{1}{\sigma_F^2}\right)_\ms$$
\end{enumerate}
\end{lem}
\begin{proof}
$a)$ These identities follow from $B\delta=[B]\delta$ for any continuous bitensor $B$, $[\sigma]=0$, $[\sigma_a]=0$ and the definition of weak derivatives.

$b)$ The first identity holds in Minkowski spacetime because $1/(8\pi^2\sigma_F)$ is the Feynman propagator of the massless vacuum state. In curved spacetimes \eqref{eq_basicidentities} imply that $(\Box+f)1/\sigma_F$ vanishes outside of the origin and thus must be a sum of derivatives of $\delta$ distributions. Because $\sigma$ depends smoothly on the metric, the coefficients in this sum must be smooth functions of the metric with appropriate mass dimension and thus $(\Box+f)1/\sigma_F=c\delta$ with a constant $c$ that can be fixed in Minkowski spacetime. 

In order to prove the second identity we recall Remark \ref{rem_fdists} and observe that it is sufficient to compute 
$$
t:=\lim_{\alpha\to 0}f(\Box+f)\frac{1}{\sigma^{1+\alpha}_F}
$$
This expression has for $y\to x$ a scaling degree $\le 4$, vanishes outside of $x=y$, depends smoothly on the metric, is covariant and has mass dimension $6$. Consequently $t=cR\delta$ where the dimensionless constant $c$ can be computed on any spacetime with $R\neq 0$. Moreover, in any spacetime where $f$ is actually continuous in a neighbourhood of the diagonal we have $t(x,y)=8\pi i f(x,x) \delta(x,y)$. A spacetime which satisfies both properties is (the patch of) de Sitter spacetime defined in conformal coordinates by the metric line element 
$$ds^2=\frac{1}{H^2\tau^2}\left(-d\tau^2 + d\vec{x}^2\right)$$
on $(-\infty,0)\times\bbR^3$, where $H$ is a constant. On this spacetime we have $R=12H^2$ and
$$\mu^2:=2 H^2 \sigma(\tau_1,\vec{x}_1,\tau_2,\vec{x}_2)=\cos^{-1}\left(\frac{\tau^2_1+\tau^2_2-(\vec{x}_1-\vec{x}_2)^2}{2\tau_1\tau_2}\right),$$
see e.g. \cite{Allen:1985ux}, where analytic continuation of $\cos^{-1}$ is understood for time--like separations. From this one can infer 
$$
\Box \sigma = 1+3 \mu \cot (\mu) \qquad \Rightarrow \qquad f = \frac{\Box \sigma-4}{\sigma} = 6H^2 \frac{\mu \cot (\mu) - 1}{\mu^2}= -\frac{R}{6} + O(\mu^2)
$$
which demonstrates that on de Sitter spacetime $f$ is continuous in a neighbourhood of the diagonal with $f(x,x)=-R/6$. 

$c)$ The distributions on both sides of each equation, considered as distributions in $y$ for fixed $x$, have the same scaling degree $<4$ for $y\to x$ and agree outside of the diagonal. Thus they agree also on the diagonal as unique extensions.

$d)$ As in the proof of a) we observe that the potential local correction term on the right hand side must be a sum of derivatives of $\delta$ with coefficients that depend smoothly on the metric because $\sigma$ does. Thus the correction term must be of the form $c\delta$ with a constant $c$ that can be computed in Minkowski spacetime. This computation may be performed by using \eqref{eq_basicidentities}, the previous statements of this lemma, and the following identities which are valid in Minkowski spacetime for any function $F$ s.t. $F(\sigma_F)$ is a distribution
$$\sigma_F\Box  F(\sigma_F)=\Box \sigma_F F(\sigma_F) - 4 F(\sigma_F) - 2\nabla_{\nabla\sigma_F}F(\sigma_F)\,,$$
$$\sigma_F\Box^2  F(\sigma_F)=\Box^2 \sigma_F F(\sigma_F) - 4 \Box F(\sigma_F) - 4\Box \nabla_{\nabla\sigma_F}F(\sigma_F)\,,$$
whereby one finds that $c=0$.
\end{proof}

These identities can be used to compute $\left(\Delta^2_F\right)_\ms$ and $\left(\Delta^3_F\right)_\ms$ in an alternative way under certain conditions.

\begin{propo}\label{prop_equivalentscheme}Let $(\M,g)$ be such that $\M$ is a normal neighbourhood and let $\Delta_F$ be a distribution on $\M^2$ of Feynman-Hadamard form \eqref{eq_DeltaF}. Then the following identities hold.
\begin{enumerate}
\item If $\Delta_F^{\alpha}$ is a well--defined distribution which is weakly meromorphic in $\alpha$, then
$$(\Delta^2_F)_\ms=\lim_{\alpha\to 0}\left(\frac{1}{M^{2\alpha}}\Delta_F^{2+\alpha}-\pp\frac{1}{M^{2\alpha}}\Delta_F^{2+\alpha}\right)+\frac{i\log(8\pi^2)}{16\pi^2}\delta\,.$$
\item If $\Delta_F^{\alpha}$ is a well--defined distribution which is weakly meromorphic in $\alpha$, then
$$(\Delta^2_F\log \left(M^{-2}\Delta_F\right))_\ms=\lim_{\alpha\to 0}\left(\frac{d}{d\alpha}\frac{1}{M^{2\alpha}}(\Delta_F)^{2+\alpha}-\pp\frac{d}{d\alpha}\frac{1}{M^{2\alpha}}\Delta_F^{2+\alpha}\right)-\frac{i\log^2(8\pi^2)}{32\pi^2}\delta\,.$$
\item If $\Delta_F^{\alpha}$ is a well--defined distribution which is weakly meromorphic in $\alpha$ and $[v]=0$, then 
$$(\Delta^3_F)_\ms=\lim_{\alpha\to 0}\left(\frac{1}{M^{2\alpha}}\Delta_F^{3+\alpha}-\pp\frac{1}{M^{2\alpha}}\Delta_F^{3+\alpha}\right)+\frac{i\left((1+2\log(8\pi^2))R+192\pi^2[w]\right)}{48(8\pi^2)^2}\delta\,.$$
\end{enumerate}
\end{propo}
\begin{proof}
$a)$ Setting $h=8\pi^2\sigma_F \Delta_F$ and $k=\sqrt{8\pi^2/h}$, we obtain
$$\frac{1}{M^{2\alpha}}\Delta_F^{2+\alpha}=\frac{h^2}{(8\pi^2)^2}\frac{1}{(Mk)^{2\alpha}}\frac{1}{\sigma_F^{2+\alpha}}\,.$$ 
Using \eqref{eq_generalexpansion}, $[h^2]=[u^2]=1$ and Lemma \ref{lem_productidentities} a), b) \& c) we may compute
\begin{align*}&\lim_{\alpha\to 0}\left(\frac{1}{M^{2\alpha}}\Delta_F^{2+\alpha}-\pp\frac{1}{M^{2\alpha}}\Delta_F^{2+\alpha}\right)\\
=\quad&\frac{h^2}{(8\pi^2)^2}\lim_{\alpha\to 0}\left(\frac{1}{(Mk)^{2\alpha}}\frac{1}{\sigma_F^{2+\alpha}}-\pp\frac{1}{(Mk)^{2\alpha}}\frac{1}{\sigma_F^{2+\alpha}}\right)\\
=\quad&\frac{h^2}{(8\pi^2)^2} \left(\left(\frac{1}{\sigma^2_F}\right)_\ms-\frac{\log (k^2)}{2}\left(\Box + f\right)\frac{1}{\sigma_F}\right)=(\Delta^2_F)_\ms-\frac{i\log(8\pi^2)}{16\pi^2}\delta\,.\end{align*}

$b)$ In analogy to $a)$, we may compute
\begin{align*}&\lim_{\alpha\to 0}\left(\frac{d}{d\alpha}\frac{1}{M^{2\alpha}}\Delta_F^{2+\alpha}-\pp\frac{d}{d\alpha}\frac{1}{M^{2\alpha}}\Delta_F^{2+\alpha}\right)\\
=\quad&\frac{h^2}{(8\pi^2)^2} \left(-\left(\frac{\log \left(M^2 \sigma_F\right)}{\sigma^2_F}\right)_\ms-\log\left(\frac{8\pi^2}{h^2}\right)\left(\frac{1}{\sigma^2_F}\right)_\ms+\frac{\log^2 \left(\frac{8\pi^2}{h^2}\right)}{4}\left(\Box + f\right)\frac{1}{\sigma_F}\right)\\
=\quad &\left(\Delta^2_F\log \left(M^{-2}\Delta_F\right)\right)_\ms+\frac{i\log^2(8\pi^2)}{32\pi^2}\delta\,.\end{align*}
 
$c)$ This can be proven in analogy to a) and b), whereby one also needs Lemma \ref{lem_productidentities} d) and the fact that $[v]=0$ implies by means of the covariant expansion of bitensors near the diagonal (see e.g. \cite[Section 5]{Poisson:2011nh}) that 
$$v=[v]+([\nabla_a v]-\nabla_a[v])\sigma^a+\R_v=[\nabla_a v]\sigma^a+\R_v\,,$$
where the remainder term $\R_v$ vanishes towards the diagonal fast than $\sigma_a$. Thus, the assumption $[v]=0$ implies that the term in $\Delta^3_F$ proportional to $\sigma^{-2}_F \log M^2\sigma_F$ does not need to be renormalised, which is crucial for the present proof. The correction term arises from the  $\log h/(8\pi^2)$ term in the expansion of $$\frac{1}{(Mk)^{2\alpha}}\frac{1}{\sigma_F^{3+\alpha}}$$ whose contribution may be computed as
$$\frac{h^3 \log\left( \frac{h}{8\pi^2}\right)}{8(8\pi^2)^3}(\Box+2f)(\Box+f)\frac{1}{\sigma_F}=-\frac{i\left(-\log(8\pi^2)\frac{R}{3}-[\Box h^3 \log (h)]\right)\delta}{8(8\pi^2)^2}=$$
$$=\frac{i\left(-\log(8\pi^2)\frac{R}{3}+[\Box u + 8\pi^2 w\Box \sigma]\right)\delta}{8(8\pi^2)^2}=\frac{i\left((1+2\log(8\pi^2))R+192\pi^2[w]\right)}{48(8\pi^2)^2}\delta\,,$$
where again Lemma \ref{lem_productidentities} a) \& b) prove to be  useful.
\end{proof}

\subsubsection{A more complicated graph}
\label{sec:complicatedgraph}
In order to show how the proposed renormalisation scheme works for graphs which have more than two vertices we discuss the renormalisation of the following triangular graph
\[
\tau_\Gamma := \Delta_{F,13}\Delta_{F,23}\Delta_{F,12}^2\,,
\]
where $\Delta_{F,ij}:=\Delta_F(x_i,x_j)$. In order to apply the forest formula \eqref{eq:forset-formula} to renormalise this graph, we note that the forests which correspond to divergent contributions are 
\begin{gather*}
\{12\}\,,\quad\{123\}\,,\quad \{12,123\}\,.
\end{gather*} 
The renormalisation of $\tau_\Gamma$ thus reads
\[
(\tau_\Gamma)_\ms = (1+R_{12}+R_{123}+R_{123}R_{12}) \tau^{(\balpha)}_\Gamma   =   (1+R_{123})(1+R_{12}) \tau^{(\balpha)}_\Gamma.
\]
In order to illustrate the explicit form of the $R$, we consider only the most singular contribution to $\tau^{(\balpha)}_\Gamma$, namely
\[
t_{\Gamma,0}^{(\balpha)} := \frac{1}{\sigma_{13}^{1+\alpha_{13}}} \frac{1}{\sigma_{12}^{2(1+\alpha_{12})}} \frac{1}{\sigma_{23}^{1+\alpha_{23}}}\,,
\]
where $\sigma_{ij} := \sigma_F(x_1,x_j)$. Note that, with obvious notation, $(8\pi^2)^{-4}u_{13} u^2_{12} u_{23} t_{\Gamma,0}$ is in fact the only contribution to $\tau_\Gamma$ which needs to be renormalised. The application of $1+R_{12}$ to $t_{\Gamma,0}^{(\balpha)}$ has already been discussed in the preceding sections and corresponds to the renormalisation of the fish graph. Indeed, after setting $\alpha_{12}$, $\alpha_{23}$ and $\alpha_{13}$ to $\alpha=\alpha_I$ for $I=\{1,2,3\}$ we obtain
\[
t_{\Gamma,1}^{(\alpha)} := \lim_{\alpha_{ij}\to\alpha}(1+R_{12}) t_{\Gamma,0}^{(\balpha)}  = \left(\left(\frac{1}{\sigma_{12}^2}\right)_\ms + O(\alpha)\right)
\frac{1}{(\sigma_{13})^{1+\alpha}} \frac{1}{(\sigma_{23})^{1+\alpha}}.
\]
The distribution $(1/\sigma^2_{12})_\ms$ is a homogeneous distribution of degree $\delta=-4$ under scaling of $x_2$ towards $x_1$, consequently, $t_{\Gamma,1}^{(\alpha)}$ has scaling degree $8+4\alpha$. 

Owing to Proposition \ref{pr:almost-homo}, we know that $t_{\Gamma,1}^{(\alpha)}$ can be decomposed into the sum of a homogeneous distribution of degree $-8-4\alpha$ and a remainder. Hence, in order to expose the poles of $t_{\Gamma,1}^{(\alpha)}$,    
we can directly apply Proposition \ref{pr:expose-poles} with $m=1$ and $c_0 = -4\alpha$. To this end, we set $u_0 := t_{\Gamma,1}^{(\alpha)}$ 
and find 
\[
u_1 := -4\alpha u_0 - E^\dagger_1 u_0   = \left(\left(\frac{1}{\sigma_{12}^{2}} \right)_\ms + O(\alpha) \right) \frac{1}{(\sigma_{13})^{1+\alpha}} \frac{1}{(\sigma_{23})^{2+\alpha}} \,G\,,
\]
where $G=G(x_1,x_2,x_3)$ is the smooth function introduced in Lemma \ref{le:rho-over-squared}. 
From \eqref{eq:expose-poles} we can infer that the principal part of $t_{\Gamma,1}^{(\alpha)}$ is
\[
\pp \,t_{\Gamma,1}^{(\alpha)} = -\frac{1}{4\alpha} \left(E^\dagger_1 +\frac{G}{\sigma_{23}}\right) \left( \left(\frac{1}{\sigma_{12}^2}\right)_\ms\frac{1}{\sigma_{13}} \frac{1}{\sigma_{23}}\right)\,,
\]
whereas the constant regular part can be easily computed as well. Consequently, the renormalised distribution
\[
(t_{\Gamma,0})_\ms = \lim_{\alpha\to 0} \left( t_{\Gamma,1}^{(\alpha)} - \pp \,t_{\Gamma,1}^{(\alpha)}   \right)
\]
can be straightforwardly computed in explicit terms.



\section{Explicit computations in cosmological spacetimes}
\label{sec_FRW}

The aim of this section is provide pr\^et-\`a-porter formulae for doing perturbative computations in the renormalisation scheme devised in the previous sections for the special case of Friedmann--Lema\^itre--Robertson--Walker (FLRW) spacetimes. We thus consider spacetimes $(\M,g)$ of the form $\M=I\times \bbR^3\subset \bbR^4$ and, in comoving coordinates,
$$g= -dt^2 + a(t)^2 d\vec{x}^2=a(\tau)^2\left(-d\tau^2+d\vec{x}^2\right).$$
Here, $t$ is cosmological time and $\tau$ is conformal time related to $t$ by $dt = a d\tau$ and 
\beq\label{eq_RFLRW} H:=\partial_t \log (a) = \frac{\partial_\tau a}{a^2}=:\frac{\H}{a}\,,\qquad R=6(\partial_t H + 2 H^2)=\frac{\partial^2_\tau a}{a^3}\,.\eeq
We consider here the spatially flat FLRW case for simplicity. Note that these spacetimes are normal neighbourhoods so that \eqref{eq_DeltaF} can be considered as a global expression and all Feynman amplitudes can be analytically regularised without the need of introducing partitions of unity such as in Remark \ref{rem:geodesicneighbourhood}.

\subsection{Propagators in Fourier space}

In comoving coordinates with conformal time, the Klein-Gordon operator reads
$$P=-\Box + \xi R + m^2=\frac{1}{a(\tau)^3}\left(\partial^2_\tau-\vec{\nabla}^2 + \left(\xi-\frac16\right)R a^2+m^2a^2\right)a(\tau).$$
It is convenient to employ Fourier transformations with respect to the spatial coordinates in order to expand quantities in QFT on FLRW spacetimes in terms of mode solutions of the free Klein-Gordon equation
$$\phi_{\vec{k}}(\tau,\vec{x})=\frac{\chi_k(\tau)e^{i\vec{k}\vec{x}}}{(2\pi)^{\frac32}a(\tau)},$$
where the temporal modes $\chi_k(\tau)$ satisfy
\begin{equation}\label{eq_modesode}
\left(\partial^2_\tau+k^2+m^2a^2 + \left(\xi-\frac16\right)R a^2\right)\chi_k(\tau)=0
\end{equation}
and the normalisation condition
\begin{equation}\label{eq_modesnormal}
{\chi_k}\partial_\tau \overline{\chi_k}-\overline{\chi_k}\partial_\tau{\chi_k}=i\,.
\end{equation}
Here, $k:= |\vec{k}|$ and $\overline{\cdot}$ denotes complex conjugation.

In particular, we can use the mode expansion in order to give explicit expressions for the various propagators of the free Klein-Gordon quantum field in a pure, Gaussian, homogeneous and isotropic state $\Omega$ (see \cite{Lueders:1990np, Pinamonti:2010is, Zschoche:2013ola} for associated technical conditions on the mode functions). To this avail, we define
\begin{equation}\label{eq_propagatorsfourier}\Delta_\sharp(x_1,x_2)=:\lim_{\epsilon\downarrow 0}\frac{1}{8\pi^3a(\tau_1)a(\tau_2)}\int_{\bbR^3} d^3k\; \widehat{\Delta_\sharp}(\tau_1,\tau_2,k)\,e^{i\vec{k}(\vec{x}_1-\vec{x}_2)-\epsilon k}\,,\end{equation}
where $\Delta_\sharp$ stands for either $\Delta_+$ (two-point function), $\Delta_{R/A}$ (retarded/advanced propagator) or $\Delta_F$ (Feynman propagator). See Section \ref{sec_propagators} for our conventions for these propagators and their relations. Recall that our renormalisation scheme preserves invariance under spacetime isometries and thus we know that renormalised powers of the Feynman propagator may also be written in the form \eqref{eq_propagatorsfourier}.

The Fourier versions of the single propagators read
\begin{gather}
\widehat{\Delta_+}(\tau_1,\tau_2,k)=\chi_k(\tau_1)\overline{\chi_k(\tau_2)}\,,\qquad \widehat{\Delta_-}(\tau_1,\tau_2,k)= \overline{\widehat{\Delta_+}(\tau_1,\tau_2,k)}\,,\notag\\
\widehat{\Delta_F}(\tau_1,\tau_2,k)=\Theta(\tau_1-\tau_2)\widehat{\Delta_+}(\tau_1,\tau_2,k)+\Theta(\tau_2-\tau_1)\widehat{\Delta_-}(\tau_1,\tau_2,k)\label{eq_propagatorsfourierexp}\,,\\
\widehat{\Delta_{R/A}}(\tau_1,\tau_2,k)=\mp i \Theta\left(\pm(\tau_1-\tau_2)\right)\left(\widehat{\Delta_+}(\tau_1,\tau_2,k)-\widehat{\Delta_-}(\tau_1,\tau_2,k)\right)\,,\notag
\end{gather}
whereas by the convolution theorem, we have the following Fourier versions of products and convolutions of multiple propagators, provided those products and convolutions are well-defined. 
Defining
\begin{gather}
\left[\Delta_{\sharp_1}\ast_4\Delta_{\sharp_2}\right](x,y):=\int_\M d^4x \sqrt{-g}\; \Delta_{\sharp_1}(x_1,x)\Delta_{\sharp_2}(x,x_2)\notag\\
\left[\widehat{\Delta_{\sharp_1}}\ast_1\widehat{\Delta_{\sharp_2}}\right](\tau_1,\tau_2,k):=\int_I d\tau \;a(\tau)^2 \,\widehat{\Delta_{\sharp_1}}(\tau_1,\tau,k)\,\widehat{\Delta_{\sharp_2}}(\tau,\tau_2,k)\label{eq_defconvolutions}\\
\left[\widehat{\Delta_{\sharp_1}}\ast_3\widehat{\Delta_{\sharp_2}}\right](\tau_1,\tau_2,k):=\int_{\bbR^3} d^3p\;\widehat{\Delta_{\sharp_1}}(\tau_1,\tau_2,p)\widehat{\Delta_{\sharp_2}}\left(\tau_1,\tau_2,|\vec{k}-\vec{p}|\right)\notag
\end{gather}
we have
\begin{gather}\label{eq_convolutionidentities}
\widehat{\prod^n_{i=1}\Delta_{\sharp_i}}(\tau_1,\tau_2,k)=\frac{1}{\left((2\pi)^3 a(\tau_1)^{2}a(\tau_2)^{2}\right)^{n-1}}\left[\widehat{\Delta_{\sharp_1}}\ast_3\cdots\ast_3\widehat{\Delta_{\sharp_n}}\right](\tau_1,\tau_2,k)\,,\\
\widehat{\Delta_{\sharp_1}\ast_4\cdots\ast_4\Delta_{\sharp_n}}=\widehat{\Delta_{\sharp_1}}\ast_1\cdots\ast_1
\widehat{\Delta_{\sharp_n}}\,.\notag
\end{gather}

Choosing a pure, Gaussian, homogeneous and isotropic state $\Omega$ of the quantized free Klein-Gordon field on a spatially flat FLRW spacetimes amounts to choosing a solution of \eqref{eq_modesode} and \eqref{eq_modesnormal} for each $k$. In order for $\Omega$ to be a Hadamard state the temporal modes $\chi_k$ have to satisfy certain conditions in the limit of large $k$ which are difficult to formulate precisely. Heuristically, a necessary but not sufficient condition is that the dominant part of $\chi_k$ for large $k$, when the mass and curvature terms in \eqref{eq_modesode} are dominated by $k^2$, is $\frac{1}{\sqrt{2k}}e^{-ik\tau}$, i.e. a positive frequency solution. Note that the retarded and advanced propagators are state-independent and thus $\widehat{\Delta_{R/A}}(\tau_1,\tau_2,k)$ is independent of the particular $\chi_k$ chosen for each $k$.

\subsection{The renormalised fish and sunset graphs in Fourier space}
\label{sec:sunfishflrw}

In perturbative calculations at low orders we encounter (pointwise) powers of $\Delta_\pm$ and $\Delta_F$. While the powers of $\Delta_\pm$ are well-defined if $\Omega$ is a Hadamard state on account of the wave front set properties of these distributions, we need to renormalise the powers of $\Delta_F$ by means of the scheme developed in the previous sections. In order to be useful for explicit computations in FLRW spacetimes, we have to develop a spatial Fourier--space version of this scheme. Having in mind the application to $\phi^4$ theory, we shall compute $\widehat{(\Delta_{F})^n_\ms}(\tau_1,\tau_2,k)$ for $n=2,3$. The difficulty in achieving this is that, to our knowledge, despite of the large symmetry of flat FLRW spacetimes, neither $\sigma$ nor the Hadamard coefficients $u$, $v$ and $w$ may written in a tractable form which can be Fourier--transformed easily. Our strategy to circumvent this problem is the following.\\\\
\noindent{\bfseries Computational strategy}

\begin{enumerate}
\item For a general mass $m$ and coupling to the scalar curvature $\xi$ and a general homogeneous and isotropic, pure and Gaussian Hadamard state $\Omega$, split $\Delta_F$ as
\beq\label{eq_propagatorsplit}\Delta_F=\Delta_{F,0}+d\,,\qquad d:=\Delta_F-\Delta_{F,0}\,,\eeq
where $\Delta_{F,0}$ must satisfy the following conditions.
\begin{itemize}
\item $\Delta_{F,0}$ is explicitly known in position space and Fourier space.
\item $\Delta_{F,0}$ is of the form
$$\Delta_{F,0}=\frac{1}{8\pi^2}\left(\frac{u_0}{\sigma_F}+v_0\log \left(M^2\sigma_F\right)\right)+w_0\,,$$
with $u_0=u$, i.e. it agrees with $\Delta_F$ in the most singular term but not necessarily in the subleading singularities.
\item $[v_0]=0$ and $\Delta^{\alpha}_{F,0}$ is weakly meromorphic in $\alpha$ such that $\left(\Delta^2_{F,0}\right)_\ms$, $\left(\Delta^2_{F,0}\log \left(M^{-2} \Delta_{F,0}\right)\right)_\ms$ and  $\left(\Delta^3_{F,0}\right)_\ms$ may be computed with Proposition \ref{prop_equivalentscheme}. This is crucial for preserving the explicit knowledge of $\Delta_{F,0}$ in position space in the renormalisation procedure, so that one may hope to compute the Fourier transforms of the renormalised powers.
\end{itemize}

\item With these assumptions on $\Delta_{F,0}$ it follows that the renormalised fish and sunset graphs may be computed as
\beq\label{eq_fishsunsetalt}(\Delta_{F})^2_\ms = (\Delta_{F,0})^2_\ms + 2 \Delta_{F,0} d+ d^2\eeq
$$(\Delta_{F})^3_\ms = (\Delta_{F,0})^3_\ms + 3 \left(\Delta^2_{F,0} d\right)_\ms+ 3  \Delta_{F,0} d^2+ d^3$$
because the non--renormalised terms in the above formulae are distributions with scaling degree $<4$ for $y\to x$ and thus can be directly and uniquely extended to the diagonal.

\item $\left(\Delta^2_{F,0}\right)_\ms$ and  $\left(\Delta^3_{F,0}\right)_\ms$ may be computed with Proposition \ref{prop_equivalentscheme} as anticipated. In order to compute $\left(\Delta^2_{F,0} d\right)_\ms$, we further split $d$ as
\beq\label{eq_dsplit}d=d_1+d_2\,,\qquad d_1:= -\frac{[v] \log \left(M^{-2} \Delta_{F,0}\right)}{8\pi^2}\,,\qquad d_2:=d-d_1\,.\eeq
Because $v=[v]+O(\sigma_a)$, $d_1$ contains the leading logarithmic singularity in $d$ (and thus $\Delta_F$) which is the only logarithmic singularity relevant for the renormalisation of the sunset graph. Consequently
\beq\label{eq_fishsunsetalt2}\left(\Delta^2_{F,0} d\right)_\ms=- \frac{[v]}{8\pi^2}\left(\Delta^2_{F,0}\log \left(M^{-2} \Delta_{F,0}\right)\right)_\ms + d_2 \left(\Delta^2_{F,0}\right)_\ms\,,\eeq
and thus Proposition \ref{prop_equivalentscheme} can be applied again.
\item Due to the symmetry of FLRW spacetimes and the assumption that the pure and Gaussian Hadamard state $\Omega$ is invariant under this symmetry, $[v]$ and $[w]$ do not depend on the spatial coordinates. Given that one succeeds to compute the spatial Fourier transforms of $\log \left(M^{-2} \Delta_{F,0}\right)$, $\left(\Delta^2_{F,0}\right)_\ms$, $\left(\Delta^2_{F,0}\log \left(M^{-2} \Delta_{F,0}\right)\right)_\ms$ and  $\left(\Delta^3_{F,0}\right)_\ms$, $\widehat{(\Delta_{F})^n_\ms}(\tau_1,\tau_2,k)$ for $n\in\{2,3\}$ may be computed by means of the convolution identities \eqref{eq_convolutionidentities}, since the Fourier transforms of $\Delta_{F,0}$, $d_1$ and $d_2$ are known by construction.
\end{enumerate}

In order to follow the computational strategy outlined above, we first compute $[v]$ and $[w]$.  Indeed, the coinciding point limit of the Hadamard coefficient $v$ reads (see e.g. \cite[Section III.1.2]{Hack:2010iw} for details)
\begin{equation}
\label{eq_coincidingV}[v]=\frac{m^2+\left(\xi-\frac16\right)R}{2}\,.
\end{equation}
Moreover, using the method of \cite{Schlemmer} to compute a spatial Fourier representation of the Hadamard parametrix $H_F$ -- here considered as \eqref{eq_DeltaF} with $w=0$ -- in FLRW spacetimes, one can compute (see the review in \cite{Degner} and a related method in \cite{Pinamonti:2010is} for the conformally coupled case)
\begin{eqnarray}
[w]&=&\lim_{x\to y}\left(\Delta_F(x,y)-H_F(x,y)\right)\notag\\
\label{eq_coincidingW}&=&\frac{1}{(2\pi)^3 a^2}\int\limits_{\bbR^3}d^3k\; |\chi_k(\tau)|^2-\frac{1}{2\sqrt{k^2+a^2m^2+a^2\left(\xi-\frac16\right)R}}\\
&&\quad +\frac{1}{16\pi^2}\left(m^2+\left(\xi-\frac16\right)R\right)\left(2\gamma-1+\log\left(
\frac{m^2+\left(\xi-\frac16\right)R}{2M^2}\right)\right)-\frac{R}{36(8\pi^2)}\notag\,,
\end{eqnarray}
where $\gamma$ is the Euler-Mascheroni constant and $H_F$ is taken with the mass scale $M$ inside of the logarithm of $\sigma$\footnote{Note that one may take instead of the function  $F(k)=1/(2\sqrt{k^2+a^2m^2+a^2\left(\xi-\frac16\right)R})$ in \eqref{eq_coincidingW} any distribution $F'(k)$ such that $F'(k)-F(k)$ is $O(k^{-5})$ for large $k$ and integrable. By taking e.g. $F'(k)=1/(2k)-\Theta(k-am)(a^2m^2+a^2\left(\xi-\frac16\right)R)/(4k^3)$ one may cancel the $\log R$ term outside of the integral.}. 

As anticipated we see that $[v]$ and $[w]$ are functions of time only (recall \eqref{eq_RFLRW}). Moreover, we see that $[v]=0$ for a conformally coupled ($\xi=\frac16$) massless scalar field. Thus, in order to pursue our computational strategy, we should look for a candidate for $\Delta_{F,0}$ among the Feynman propagators in suitable states of this theory. In fact, choosing the conformal vacuum state of the massless conformally coupled scalar field does the job. The conformal vacuum is given by choosing the modes $\chi_k(\tau)=e^{-ik\tau}/\sqrt{2k}$, and thus the Feynman propagator $\Delta_{F,0}$ in this state is of the form
\begin{equation}\label{eq_propagatorconformal}
\Delta_{F,0}(x_1,x_2)=\frac{1}{8\pi^2 a(\tau_1)a(\tau_2)}\frac{1}{\sigma_{F,\bbM}(x_1,x_2)}\,,\qquad \widehat{\Delta_{F,0}}(\tau_1,\tau_2,k)=\frac{e^{-ik|\tau_1-\tau_2|}}{2k}\,.
\end{equation}
Here, and in the following, the index $_{\bbM}$ indicates quantities in Minkowski spacetime, in particular $\sigma_{\bbM}(x_1,x_2)=\frac12(\vec{x}_1-\vec{x_2})^2-\frac12(\tau_1-\tau_2)^2$. $\Delta^\alpha_{F,0}$ is weakly meromorphic in $\alpha$ because the massless vacuum Feynman propagator in Minkowski spacetime has this property and the conformal rescaling by $a$ does not violate it. Thus, we may follow our computational strategy and compute $\left(\Delta^2_{F,0}\right)_\ms$, $\left(\Delta^2_{F,0}\log \left(M^{-2} \Delta_{F,0}\right)\right)_\ms$ and  $\left(\Delta^3_{F,0}\right)_\ms$ by means of Proposition \ref{prop_equivalentscheme}. This is easily done using \eqref{eq_generalexpansion} for $\sigma_{F,\bbM}$ rather than $\sigma_F$ and $h=\sqrt{8\pi^2 a(\tau_1)a(\tau_2)}=\sqrt{8\pi^2 a\otimes a}$. The results are
\begin{align*}
(\Delta_{F,0})^2_\text{ms}&= \lim_{\alpha\to 0}\left(\frac{1}{M^{2\alpha}}(\Delta_{F,0})^{2+\alpha}-\pp\frac{1}{M^{2\alpha}}(\Delta_{F,0})^{2+\alpha}\right)+\frac{i\log(8\pi^2)}{16\pi^2}\delta\\
&= \lim_{\alpha\to 0}\frac{1}{(8\pi^2)^2a^2\otimes a^2}\left(\frac{1}{(M\sqrt{8\pi^2 a\otimes a})^{2\alpha}}\frac{1}{\sigma_{F,\bbM}^{2+\alpha}}-\pp \frac{1}{(M\sqrt{8\pi^2 a\otimes a})^{2\alpha}}\frac{1}{\sigma_{F,\bbM}^{2+\alpha}}\right)+\frac{i\log(8\pi^2)}{16\pi^2}\delta\\
&=-\frac{1+2\log (a)}{16\pi^2 a^4}i\delta_\bbM-\frac{1}{2(8\pi^2)^2 a^2\otimes a^2}\Box_{\bbM}\frac{\log\left(M^2\sigma_{F,\bbM}\right)}{\sigma_{F,\bbM}}\,,
\end{align*}
\begin{align*}
\left(\Delta^2_{F,0}\log \left(M^{-2} \Delta_{F,0}\right)\right)_\text{ms}&= \lim_{\alpha\to 0}\left(\frac{d}{d\alpha}\frac{1}{M^{2\alpha}}(\Delta_{F,0})^{2+\alpha}-\pp\frac{d}{d\alpha}\frac{1}{M^{2\alpha}}(\Delta_{F,0})^{2+\alpha}\right)-\frac{i\log^2(8\pi^2)}{32\pi^2}\delta\\
&=\frac{2+2\log (a^2 8\pi^2)+\log^2 (a^2)}{32\pi^2 a^4}i\delta_\bbM+\frac{1}{4(8\pi^2)^2 a^2\otimes a^2}\Box_{\bbM}\frac{\log^2\left(M^2\sigma_{F,\bbM}\right)}{\sigma_{F,\bbM}}\\
&\quad +\frac{1+\log(8\pi^2) a\otimes a}{2(8\pi^2)^2 a^2\otimes a^2}\Box_{\bbM}\frac{\log\left(M^2\sigma_{F,\bbM}\right)}{\sigma_{F,\bbM}} \,,
\end{align*}
and
\begin{align*}
(\Delta_{F,0})^3_\text{ms}&= \lim_{\alpha\to 0}\left(\frac{1}{M^{2\alpha}}(\Delta_{F,0})^{3+\alpha}-\pp\frac{1}{M^{2\alpha}}(\Delta_{F,0})^{3+\alpha}\right)+\frac{i\left((1+2\log(8\pi^2))R+192\pi^2[w]\right)}{48(8\pi^2)^2}\delta\\
&=-\frac{(15+12\log (a))\Box_\bbM+6(\Box_\bbM \log (a))+2(\partial^2_\tau a)/a}{48(8\pi^2)^2a^6}i\delta_\bbM-\frac{1}{8(8\pi^2)^3 a^3\otimes a^3}\Box^2_{\bbM}\frac{\log\left(M^2\sigma_{F,\bbM}\right)}{\sigma_{F,\bbM}}\,.
\end{align*}
where we have used $\delta=\delta_{\bbM}/a^4$, $f_{\bbM}=0$ and the fact that, by \eqref{eq_coincidingW}, $8\pi^2[w_0]=-R/36$ for the conformal vacuum state of the massless, conformally coupled scalar field.

Using these results as well as the Fourier representation of $1/\sigma_{F,\bbM}$ \eqref{eq_propagatorconformal} and $\log \left(M^2\sigma_{F,\bbM}\right)$
\eqref{eq_flog}, and convolution identities, we can finally obtain the Fourier versions of the renormalised powers of $\Delta_{F,0}$. For instance, we find for $\left(\Delta^2_{F,0}\right)_\ms$
\begin{gather}\label{eq_fouriersquare}
\widehat{\left(\Delta^2_{F,0}\right)_\text{ms}}(\tau_1,\tau_2,k)=-\frac{1+2\log (a(\tau_1))}{16a(\tau_1)^2\pi^2}\delta(\tau_1-\tau_2)-\\
-\frac{1}{16\pi^3 a(\tau_1)a(\tau_2)}(\partial^2_{\tau_1}+k^2)\int_{\bbR^3}d^3p\,\left(\frac12\left(\frac{1}{p^3}\right)_{\text{ren},M}+\frac{i|\tau_1-\tau_2|}{2p^2}\right)\frac{1}{2|\vec{k}-\vec{p}|}e^{-i(p+|\vec{k}-\vec{p}|)|\tau_1-\tau_2|},\notag
\end{gather}
where the appearing renormalisation of $1/p^3$ is defined in \eqref{eq_regk3}. Note that the $\vec{p}$-integral has no convergence problems for large $p$ because one may write the potentially dangerous $-i|\tau_1-\tau_2|e^{-2ip|\tau_1-\tau_2|}/p$ contribution as $\partial_p( e^{-2ip|\tau_1-\tau_2|}/(2p^2))$ plus an $O(p^{-3})$ term. Regarding the convergence for small $p$ we observe that the integral is manifestly convergent if $k\neq0$, thus yielding a well-defined distribution in $\vec{k}$ on $\bbR^3\setminus\{0\}$. The scaling degree of this distribution is easily seen to be $1<3$ and thus a unique extension towards the origin exists. In practical terms this means that the integral for $k=0$ may be computed as a limit $k\to0$ of the integral with nonvanishing $k$ without any renormalisation.

\subsection{Example: the two-point function for a quartic potential up to second order}

In order to compute the analytic expressions corresponding to the graphs in Figure \ref{fig_2pf}, we may use the Fourier versions of the appearing propagators \eqref{eq_propagatorsfourierexp}, \eqref{eq_fouriersquare}, and the analogous expressions for $\widehat{\left(\Delta^2_{F,0}\log\left(M^{-2}\Delta^2_{F,0}\right) \right)_\text{ms}}(\tau_1,\tau_2,k)$ and 
$\widehat{\left(\Delta^3_{F,0}\right)_\text{ms}}(\tau_1,\tau_2,k)$ the explicit form of $\mu(x)=3\lambda w(x,x)$ in \eqref{eq_coincidingW}, as well as \eqref{eq_fishsunsetalt}, \eqref{eq_fishsunsetalt2} and the identities for products and convolutions \eqref{eq_defconvolutions}, \eqref{eq_convolutionidentities}. Note that $\mu(x)$ is in fact only time-dependent because $\Omega$ was chosen homogeneous and isotropic. Thus the integrals with $\mu$-vertices can be computed partly with the above-mentioned identities by means of 
$$\widehat{(1\otimes \mu) \Delta_{\sharp}}(\tau_1,\tau_2,k)=\mu(\tau_2)\widehat{\Delta_{\sharp}}(\tau_1,\tau_2,k),\qquad\widehat{(\mu\otimes 1) \Delta_{\sharp}}(\tau_1,\tau_2,k)=\mu(\tau_1)\widehat{\Delta_{\sharp}}(\tau_1,\tau_2,k).$$
Similarly, the bubbles in the third line of Figure \ref{fig_2pf} contribute only time-dependent vertex factors which can be computed as 
$$h_\sharp(\tau):=\int_\M d\tau_1 d^3x_1\; a(\tau_1)^4\mu(\tau_1)\Delta_\sharp(\tau,\tau_1,\vec{x}-\vec{x}_1)=\frac{1}{a(\tau)}\int_I d\tau_1\;a(\tau_1)^3 \mu(\tau_1)\widehat{\Delta_\sharp}(\tau,\tau_1,0)$$
where $\Delta_\sharp$ is either $\Delta^2_+$ or $\left(\Delta^2_F\right)_\ms$.

With these preparations, we can compute e.g. the first graphs of the fourth and fifth line in Figure \ref{fig_2pf} 
in Fourier space as
\begin{align*}
\widehat{\Delta_R\ast_4((h_F\otimes 1)\Delta_+)}&=\widehat{\Delta_R}\ast_1\widehat{\left((h_F\otimes 1)\Delta_+)\right)}\\
&=\int_{I^2} d\tau_3\,d\tau_4\; a(\tau_3)a(\tau_4)^3\mu(\tau_4)\widehat{\Delta_R}(\tau_1,\tau_3,k)\widehat{\Delta_+}(\tau_3,\tau_2,k)\widehat{(\Delta^2_F)_\text{ms}}(\tau_3,\tau_4,0)
\end{align*}
and
\begin{align*}
\widehat{\Delta_R\ast_4(\Delta_F)^3_\text{ms}\ast_4\Delta_+}&=\widehat{\Delta_R}\ast_1\widehat{(\Delta_F)^3_\text{ms}}\ast_1\widehat{\Delta_+}\\
&=\int_{I^2} d\tau_3\,d\tau_4\; a(\tau_3)^2a(\tau_4)^2\widehat{\Delta_R}(\tau_1,\tau_3,k)\widehat{(\Delta^3_F)_\text{ms}}(\tau_3,\tau_4,k)\widehat{\Delta_+}(\tau_4,\tau_2,k).
\end{align*}

\subsection{More complicated graphs on cosmological spacetimes}

In order to compute the Fourier transforms of more complicated graphs on FLRW spacetimes, one can use a strategy generalising the one employed in Section \ref{sec:sunfishflrw}. Namely, one again decomposes the Feynman propagator $\Delta_F$ into several pieces which capture the relevant singularities and can be expressed in terms of the conformal vacuum Feynman propagator $\Delta_{F,0}$ whose explicit form in position and Fourier space is well--known in contrast to the form of $\sigma$ itself. The corresponding decomposition of general Feynman amplitudes $\tau_\Gamma$ is straightforward. The only non--trivial step is to generalise Proposition \ref{prop_equivalentscheme} to the case of general amplitudes, i.e. to compute the difference between the minimal subtraction scheme used in conjunction with either analytically regularising powers of $\sigma$ directly or analytically regularising powers of the full propagator $\Delta_{F,0}$. However, we do not foresee any problems in obtaining such a generalisation by proving versions of Lemma \ref{le:rho-over-squared} and Proposition \ref{pr:almost-homo} for $\Delta_{F,0}$ rather than $\sigma$.

In fact, one can also skip this last step by taking a rather pragmatic approach and working directly with the renormalisation scheme consisting of decomposition in $\Delta_{F,0}$, analytic regularisation of powers of this propagator and minimal subtraction of the principal parts. This scheme, clearly applicable only to conformally flat spacetimes, satisfies all properties proved in Proposition \ref{pr:propertiesscheme}, with two exceptions. It is not obvious whether the Principle of Perturbative agreement with respect to generalised mass perturbations holds for this scheme, whereas locality and covariance of course only hold in the sense restricted to conformally flat spacetimes. In this respect it is essential that the Feynman propagator of the conformal vacuum $\Delta_{F,0}$ on conformally flat spacetimes is manifestly ``geometric'', because the corresponding propagator of the massless Minkowski vacuum has this property.


\section{Summary and outlook} 

In this work, we have introduced a renormalisation scheme on curved spacetimes consisting of a particular analytic regularisation of the Feynman propagator, and thus of all Feynman diagrams, and a minimal subtraction of the principal (pole) part of the resulting meromorphic expressions. We have argued that this scheme has all properties that a physically meaningful renormalisation scheme on curved spacetimes should have and that it is in fact a particular form of differential renormalisation. The renormalisation scheme discussed in this work has the advantage that it is 
\begin{enumerate}
\item directly applicable to spacetimes with Lorentzian signature, 
\item manifestly (local and) covariant, 
\item manifestly invariant under any spacetime isometries present,
\item capturing correctly the non--geometric and non--unique state--dependent contribution of Feynman amplitudes and not only the geometric divergent part, which is unique up to finite renormalisations,
\item well--suited for practical computations, e.g. in cosmological spacetimes,
\item constructed to all orders in perturbation theory,
\item and mathematically rigorous.
\end{enumerate}
To the best of our knowledge, other renormalisation schemes on curved spacetimes discussed in the literature such as dimensional regularisation, local momentum space methods, zeta--function regularisation, heat--kernel techniques, generic Epstein--Glaser renormalisation and, on cosmological spacetimes, dimensional regularisation only with respect to spatial variables, lack at least one of the features listed above.

In order to demonstrate the practical applicability of the scheme, we have computed several examples on generic curved spacetimes. Moreover, we have shown how explicit computations in cosmological spacetimes can be done, in particular, how the renormalisation scheme initially defined in position space can be interpreted in terms of quantities Fourier transformed with respect to comoving spatial coordinates.

We have discussed the renormalisation scheme only for scalar fields in four spacetime dimensions, however, the extension to other spacetime dimensions is straightforward. Moreover, as the analytic regularisation discussed in this work consists of regularising only inverse powers of the squared geodesic distance, it can be straightforwardly generalised to field theories with higher spin, with and without gauge--invariance. In particular, spinorial quantities can be directly regularised without the need to worry about their dependence on the dimension such as in dimensional regularisation. Finally, we expect that a generalisation of the scheme introduced in this work to gauge theories yields a scheme which preserves or can be modified to preserve the local gauge symmetry in theories which are free of anomalies.


\begin{acknowledgments}
The authors would like to thank Klaus Fredenhagen and Markus Fr\"ob for interesting discussions. The work of T.-P.H. has been supported by a Research Fellowship of Deutsche Forschungsgemeinschaft (DFG).
\end{acknowledgments}

\appendix

\section{Conventions and computational details}

\subsection{Propagators of the free Klein--Gordon field and their relations}
\label{sec_propagators}
$$\Delta_+(x,y)=\langle \phi(x)\phi(y)\rangle_\Omega = \langle \phi(x)\star\phi(y)\rangle_\Omega=\overline{\langle \phi(y)\phi(x)\rangle_\Omega}=\overline{\Delta_-(x,y)}$$
$$\Delta(x,y)=\frac{1}{i}\left(\Delta_+(x,y)-\Delta_-(x,y)\right)=\Delta_R(x,y)-\Delta_A(x,y)$$
$$\Delta_{R/A}=\pm\Theta(\pm(t_x-t_y))\Delta(x,y)$$
$$\Delta_F(x,y)=\langle T\left(\phi(x)\phi(y)\right)\rangle_\Omega = \langle \phi(x)\cdot_T\phi(y)\rangle_\Omega=\Theta(t_x-t_y)\Delta_+(x,y)+\Theta(t_y-t_x)\Delta_-(x,y)$$
\begin{eqnarray*}
\Rightarrow\qquad \Delta_F(x,y)&=&\frac{1}{2}\left(\Delta_+(x,y)+\Delta_-(x,y)\right)+\frac{i}{2}\left(\Delta_R(x,y)+\Delta_A(x,y)\right)\\
&=&\Delta_+(x,y)+i\Delta_A(x,y)\\
&=&\Delta_-(x,y)+i\Delta_R(x,y)
\end{eqnarray*}


\subsection{Fourier transform of the logarithmic term on Minkowski spacetime}
In order to compute the Fourier-transform of $\log\left( M^2 \sigma_{F,\bbM}\right)$, we recall that the Feynman-propagator of the Klein-Gordon field with mass $m$ in the Minkowski vacuum is given by
\begin{align*}
\Delta_{F,m,\bbM}&=\lim_{\epsilon\downarrow 0}\frac{1}{(2\pi)^3}\int_{\bbR^3}d^3k\, \frac{e^{-i\sqrt{k^2+m^2}|\tau_1-\tau_2|}\,e^{i\vec{k}\left(\vec{x}-\vec{y}\right)}\,e^{-\epsilon k}}{2\sqrt{k^2+m^2}}\\
&=\frac{1}{8\pi^2}\sqrt{\frac{2m^2}{\sigma_{F,\bbM}}}K_1\left(\sqrt{2 m^2\sigma_{F,\bbM}}\right)\\
&=\frac{1}{8\pi^2}\left(\frac{1}{\sigma_{F,\bbM}}+\frac{m^2}{2}\left(1+\frac{m^2 \sigma_{F,\bbM}}{4}\right)\log\left(\frac{e^{2\gamma}m^2\sigma_{F,\bbM}}{2}\right)-\frac{m^2}{2}\left(1+\frac{5m^2\sigma_{F,\bbM}}{8}\right)\right)+O(m^4),
\end{align*}
where $K_1$  is a modified Bessel function and $\gamma$ is the Euler-Mascheroni constant. Using this, we find
\begin{align*}
&\log \left(M^2 \sigma_{F,\bbM}\right)\\
&=\lim_{m\to 0}\left(16\pi^2 \frac{d\Delta_{F,m,\bbM}}{d \,m^2} - 
\log\left(\frac{e^{2\gamma}m^2}{2 M^2}\right)\right)\\
&=-\lim_{m\to 0}\left(\lim_{\epsilon\downarrow 0}\frac{1}{\pi}\int_{\bbR^3}d^3 k \frac{1+i\sqrt{k^2+m^2}|\tau_1-\tau_2|}{2(k^2+m^2)^{\frac32}}  e^{-i\sqrt{k^2+m^2}|\tau_1-\tau_2|}e^{i\vec{k}\left(\vec{x}-\vec{y}\right)}\,e^{-\epsilon k}+ \log\left(\frac{e^{2\gamma}m^2}{2 M^2}\right)\right)\\
&=-\lim_{\epsilon\downarrow 0}\frac{1}{\pi}\int_{\bbR^3}d^3 k\left(\lim_{m\to 0}\left(\frac{1}{2(k^2+m^2)^{\frac32}}+\pi\log\left(\frac{e^{2\gamma}m^2}{2 M^2}\right)\delta(\vec{k})\right)\right.\\
&\qquad\qquad\qquad \left.+\frac{i|\tau_1-\tau_2|}{2k^2}\right)e^{-ik|\tau_1-\tau_2|}e^{i\vec{k}\left(\vec{x}-\vec{y}\right)}\,e^{-\epsilon k}\\
&=-\lim_{\epsilon\downarrow 0}\frac{1}{\pi}\int_{\bbR^3}d^3 k\left(\frac{1}{2}\left(\frac{1}{k^3}\right)_{\text{ren},M}+\frac{i|\tau_1-\tau_2|}{2k^2}\right)e^{-ik|\tau_1-\tau_2|}e^{i\vec{k}\left(\vec{x}-\vec{y}\right)}\,e^{-\epsilon k}\\
&=\frac{1}{(2\pi)^{\frac32}}\lim_{\epsilon\downarrow 0}\int_{\bbR^3}d^3 k \;
\text{flog}(\tau_1-\tau_2,k)\,e^{i\vec{k}\left(\vec{x}-\vec{y}\right)}\,e^{-\epsilon k}
\end{align*}
where the appearing renormalisation of the (tempered) distribution $1/k^3$ is 
\begin{equation}\label{eq_regk3}\left(\frac{1}{k^3}\right)_{\text{ren},M}:=\lim_{m\to 0}\left(\frac{1}{(k^2+m^2)^{\frac32}}+\pi\log\left(\frac{e^{4\gamma}m^4}{4 M^4}\right)\delta(\vec{k})\right)\end{equation}
and
\begin{equation}\label{eq_flog}
\text{flog}(\tau_1-\tau_2,k):=-\sqrt{8\pi}\left(\frac{1}{2}\left(\frac{1}{k^3}\right)_{\text{ren},M}+\frac{i|\tau_1-\tau_2|}{2k^2}\right)e^{-ik|\tau_1-\tau_2|}
\end{equation}
is the sought-for spatial Fourier transform of $\log \left(M^2 \sigma_{F,\bbM}\right)$.


\end{document}